\documentclass[12pt,final,
thmsa,
epsf,a4paper]
{article}
\newtheorem{theorem}{Theorem}

\newtheorem{corollary}{Corollary}
\newtheorem{definition}{Definition}
\newtheorem{example}{Example}
\newtheorem{lemma}{Lemma}
\newtheorem{remark}{Remark}

\newenvironment{proof}[1][Proof]{\emph{#1.} }{\  \hfill $\square $ \vspace{5 pt}}

\usepackage{stackrel}
\usepackage{amssymb}
\usepackage[dvips]{graphics}
\usepackage{amsmath}

\usepackage{mathpazo}

\usepackage{enumerate}

\usepackage{amsfonts}

\usepackage{amssymb}

\usepackage{enumerate}

\usepackage{mathrsfs}

\usepackage[usenames]{color}

\usepackage{pgf,tikz}

\usetikzlibrary{decorations.pathreplacing,angles,quotes}
\usetikzlibrary{arrows.meta}
\usetikzlibrary{arrows, decorations.markings}

\usepackage{hyperref}
\hypersetup{
    colorlinks,
    citecolor=blue,
    filecolor=black,
    linkcolor=blue,
    urlcolor=blue
}

\usepackage{pgf,tikz}
\usetikzlibrary{arrows}

\usepackage[utf8]{inputenc}
\usepackage{amssymb,geometry}
\usepackage{fancyhdr}

\usepackage{soul}

\usepackage{emptypage}

\begin{document}

\author{Pablo A. Neme \and Jorge Oviedo}

\title{ON THE SET OF MANY-TO-ONE STRONGLY STABLE FRACTIONAL MATCHINGS}
\maketitle

\begin{abstract}For a many-to-one matching market where firms have strict and $\boldsymbol{q}$-responsive preferences, we give a characterization of the set of  strongly stable fractional matchings as the union of the convex hull of all connected sets of stable matchings. Also, we prove that a strongly stable fractional matching is represented as a convex combination of stable matchings that are ordered in the common preferences of all firms.
\end{abstract}

\noindent\textit{MSC2000 subject classification:} Primary: 90C05; secondary: 91B68.\\
\maketitle
\section{Introduction.}

A large part of the matching literature studies many-to-one matching markets. The agents in these markets are divided into two disjoint sets: The \textit{many}-side of the market, namely resident doctors, students, workers, etc, and the \textit{one}-side, namely hospitals, colleges, firms, etc.
The main property studied in the matching literature is stability. A matching is called stable if all agents have acceptable partners and there is no unmatched pair (hospital-doctor, college-student, firm-worker, etc.), where both agents would prefer to be matched to each other rather than staying with their current partners under the proposed matching. 
Each agent has a preference list that determines an order over the agents or sets of agents on the other side of the market, with the possibility of staying unmatched. In this paper, the agents on the many-side have $\boldsymbol{q}$-responsive and strict preferences. 

Linear programming is a widely used mathematical tool in matching theory. Each matching can be
represented by an assignment matrix called the \emph{incidence
vector} of the matching. 

 Vande Vate \cite{vate1989linear} and Rothblum \cite{rothblum1992characterization} present a system of linear inequalities that characterizes the set of stable matchings of the marriage market for two different restrictions of the market. Both papers show that the set of stable matchings for the marriage market
corresponds to the set of incidence vectors (integer solutions for linear inequalities). In other words, stable matchings are
exactly the extreme points of the polytope generated by the system of linear
inequalities.
Roth \textit{et al.} \cite{roth1993stable}, for the marriage market, introduce a linear
program that characterizes all stable matchings as the integer solutions.

Linear programming approaches have been developed for the theory of stable matching markets also by Abeledo and Rothblum \cite{abeledo1995courtship} \cite{abeledo1994stable}, Abeledo and Blum \cite{abeledo1996stable}, Abeledo \textit{et al.} \cite{abeledo1996canonical}, Fleiner \cite{fleiner2001matroid}, \cite{fleiner2003stable}, Sethuraman \textit{et al.} \cite{sethuraman2006many}, Sethuraman and Teo \cite{teo1998geometry}, and many others.

Ba\"{i}ou and Balinski \cite{baiou2000stable} present two characterizations of the convex polytope for the many-to-one matching market. We focus on one of these characterizations (the most general one). 

Lotteries over stable matchings have been studied in many instances in the literature. For the marriage market, Roth \textit{et al.} \cite{roth1993stable} studied lotteries over stable matchings via linear programming. 
When the extreme points of the convex polytope generated by the constraints of a linear program are exactly the stable matchings of the market (this is the case, for instance, in the marriage market) a random matching coincides with the concept of stable fractional matching. Roth \textit{et al.} \cite{roth1993stable} define a \textit{stable fractional matching} as a not necessarily integer solution of the linear program. When the extreme points are not all integer, these two concepts are not the same, for instance, in a many-to-one matching market with $\boldsymbol{q}$-responsive and strict preferences. That is to say, a random matching is always a stable fractional matching, but some stable fractional matchings cannot be written as a lottery over stable matchings. Example 1 shows a many-to-one matching market with an extreme point that is not  a stable matching.

Each entry of an incidence vector of a stable fractional matching can be interpreted as the time that each agent spends with one agent on the other side of the market.  
For a stable fractional matching,  it can happen that two agents, one of each side of the market, have an incentive to increase the time that they spend together at the expense of those matched agents that they like less than each other at a stable fractional matching. To study a ``good'' fractional solution, the idea is to avoid this and prevent that agents have incentives to ``block" the stable fractional matching in a fractional way. For a marriage market, Roth \textit{et al.} \cite{roth1993stable} define a \textit{strongly stable fractional matching} as a stable fractional matching that fulfils non-linear equalities  that represent this non-blocking condition mentioned above. In other words, a stable fractional matching that fulfils the non-linear equalities from Roth \textit{et al.} \cite{roth1993stable}, is a strongly stable fractional matching.
Neme and Oviedo \cite{neme2019characterization} give a characterization of the strongly stable fractional matching for the marriage market. Our work extends their result and provides a characterization for the set of many-to-one strongly stable fractional matchings. We extend the strong stability condition from  Roth \textit{et al.} \cite{roth1993stable} to a many-to-one matching market. Our first result states that a strongly stable fractional matching is represented by a convex combination among stable matching that are ordered in the eyes of all firms (Theorem \ref{sucesion decreciente}). Since we focus on one of the characterizations due to Ba\"{i}ou and Balinski \cite{baiou2000stable}, a salient question  now is, are the non-integer extreme points of this convex polytope strongly stable fractional matchings? A corollary of Theorem \ref{sucesion decreciente} answers negatively this question.

In the school choice set-up, strong stability for lotteries has been introduced by Kesten and \"{U}nver \cite{kesten2015theory}, they called it ex-ante stability for lotteries. In this market, they deal with indifferences in the priority of the schools.  Kesten and \"{U}nver \cite{kesten2015theory} also present a fractional deferred-acceptance algorithm that computes a unique strongly ex-ante stable random matching. Their paper analyses the strategy proofness and efficiency of this mechanism. Our characterization goes in another direction, we study the relationship among the stable matchings that are involved in the lotteries.

Bansal \textit{et al.} \cite{bansal2007polynomial} and Cheng \textit{et al.} \cite{cheng2008unified} study the concept of cycles in preferences and cyclic matchings for  many-to-many and many-to-one matching markets, respectively. These papers are an extension of Irving and Leather \cite{irving1986complexity}.
To seek for cycles in preferences, these authors first reduce
the preference lists of all agents. We present the reduction procedure for our market in the Appendix.
This reduction procedure allows us to find cycles in preferences.  Since the cycles of a
reduced list are disjoint, we extend the definition of cyclic matching to a set of cycles in the reduced preference profile.  

Following the extension of cyclic matching used by Bansal \textit{et al.} \cite{bansal2007polynomial} and Cheng \textit{et al.} \cite{cheng2008unified}, we define a \textit{connected set} generated by a stable matching $\mu$ as the set of all cyclic matchings of $\mu$ (including $\mu$). Then, we characterize a strongly stable fractional matching as a lottery over stable matchings that belong to the same connected set (Theorem \ref{caraterizacion}). Moreover, by  Theorem \ref{sucesion decreciente}, we prove that the stable matchings that belong to the same connected set, also have the decreasing order in the eyes of all firms.  In this way, we characterize the set of all strongly stable fractional matchings as the union of the convex hulls of these connected sets (Corrolary \ref{corolario 2}).

Roth \textit{et al.} \cite{roth1993stable}, (in Corollary 21) proved a necessarily condition that states that in a strongly stable fractional matching, each agent is matched with at most two agents of the other side of the market. Schlegel \cite{schlegel2018note} generalizes this necessarily condition for the school choice set-up with strict priorities (similar setting than ours). He shows that a strongly stable fractional matching fulfils that each worker has a positive probability to be matched to at most two distinct firms, and for each firm, all but possibly one position are assigned deterministically. For the one position that is assigned by a lottery, two workers have a positive probability of been matched to the firm (here stated as Corrolary \ref{corolario 3}). Further, although he proves that a strongly stable fractional matching is ``almost" integral, he does not describes which agents are matched (there are several ``almost" integral stable fractional matchings that are not strongly stable, Example \ref{ejemplo b-b} presents an ``almost" integral stable fractional matching that is not strongly stable).
Recall that our characterization gives a necessary and sufficient condition for a stable fractional matching to be strongly stable. As a particular case, our characterization gives an alternative proof for these two results, for the school choice set-up due to Schlegel \cite{schlegel2018note} is straightforward, and for the marriage market due to Roth \textit{et al.} \cite{roth1993stable}, it's necessary only to set all quotas of all firms equal to one. Moreover, our characterization shows explicitly which are the matched agents in a strongly stable fractional matching, through the stable matching involved in the convex combination (\textit{Cf.} (\ref{para el corolario}) in proof of Theorem \ref{caraterizacion}).

This paper is organized as follows. Section 2, we formally introduce the market, preliminary results, and one of Ba\"{i}ou and Balinski's  characterizations of stable matchings. Section 3, we  define a strongly stable fractional matching and  prove that it can be represented by a convex combination over stable matchings that are ordered for all firms. We also discuss cycles and cyclic matching properties that we use in the characterization result. In Section 4, we present our characterization of a strongly stable fractional matching. The Appendix contains the reduction procedure, lemmas, and proofs of the lemmas needed for our characterization.

\section{Preliminary Results.}

The many-to-one matching market that we study, consists of two sets of agents, the set
of firms $F=\{f_{1},\ldots,f_{n}\}$ and the set of workers $W=\{w_{1},\ldots,w_{m}\}$. Each worker $w$ has an antisymmetric,  transitive, and complete preference relation $\succ_{w}$ over $F\cup\{w\}$, and each firm $f$ has an antisymmetric, transitive, and complete preference relation $\succ_{f}$ over the power set of workers, $2^{W}$. Also, each firm $f$ has a maximum number of positions to fill: its quota, denoted by $q_{f}$. Let $\boldsymbol{q}=(q_f)_{f \in F}$ be the vector of quotas. Given $W_{0},W_{1}\subseteq{W}$, we write $W_{0}\succeq_{f}W_{1}$ to indicate that the firm $f$ likes $W_{0}$ as much as $W_{1}$. Given the preference relation $\succ_{f}$, we say that $W_{0}\succ_{f}W_{1}$ when $W_{0}\succeq_{f}W_{1}$ and $W_{0}\neq{W_{1}}$. Analogously, for each worker $w$, and any two firms, $f_{0},f_{1}\in{F}$, we write $f_{0}\succeq_{w}f_{1}$ and $f_{0}\succ_{w}f_{1}$.

\textit{Preference profiles} are $(n+m)-tuples$ of preference relations and they are denoted by $P=(\succ_{f_{1}},\ldots,\succ_{f_{n}},\succ_{w_{1}},\ldots,\succ_{w_{m}})$. The matching market for the sets $W$ and $F$ with the preference profile $P$ and vector of quotas $\boldsymbol{q}$ is denoted by $(F,W,P,\boldsymbol{q})$.

We say that a pair $(f,w)\in{F\times{W}}$ is an \textit{acceptable pair} at $P$ if $w$ is acceptable for $f$, and $f$ is acceptable for $w$, that is, $\{w\}\succ_{f}\emptyset$ and $f\succ_{w}w$. Let us denote by $A(P)$ the set of all acceptable pairs of the matching market $(F,W,P,\boldsymbol{q})$, (simply $A$, when no confusion arises).

The \textit{assignment problem} consists of matching workers with firms keeping the bilateral nature of their relationship and allowing for the possibility that firms and workers remain unmatched. Formally,

\begin{definition}
Let $(F,W,P,\boldsymbol{q})$ be a many-to-one matching market. A \textbf{matching} $\mu$ is a mapping from the set $F \cup W$ into the set of all subsets of $F \cup W$ such that, for all $w\in{W}$ and $f\in{F}$:
\begin{enumerate}
\item $\vert\mu(w)\vert=1$ and if $\mu(w)\neq\{w\}$,  then $\mu(w)\subseteq{F}$.
\item $\mu(f)\in{2^{W}}$ and $\vert\mu(f)\vert\leq q_{f}$.
\item $\mu(w)=\{f\}$ if and only if $w\in\mu(f)$.
\end{enumerate}
\end{definition}

Usually we will omit the curly brackets, for instance, instead of condition 1. and 3., we will write: ``1. $\vert\mu(w)\vert=1$ and if $\mu(w)\neq w$,  then $\mu(w)\subseteq{F}$.'' and ``3. $\mu(w)=f$ if and only if $w\in\mu(f)$.''  
Assume that each firm $f$ gives its ranking of workers individually, and orders subsets of workers in a \textit{responsive} manner. That is to say, adding ``good" workers to a set leads to a better set, whereas adding ``bad" workers to a set leads to a worse set. In addition, for any two subsets that differ in only one worker, the firm prefers the subset containing the most preferred worker. Formally,
\begin{definition}
The preference relation $\succ_f$ over $2^{W}$ is \textbf{$\boldsymbol{q}$-responsive} if it satisfies the following conditions:
\begin{enumerate}
\item For all $T\subseteq{W}$ such that $\vert{T}\vert>q_{f}$, we have that $\emptyset\succ_{f}T$.
\item For all $T\subseteq{W}$ such that $\vert{T}\vert<q_{f}$ and $w\not\in{T}$, we have that $$T\cup\{w\}\succ_{f}T \mbox{ if and only if } w\succ_{f}\emptyset.$$
\item For all $T\subseteq{W}$ such that $\vert T  \vert<q_{f}$ and $w,w'\not\in{T}$, we have 
$$T \cup \{w\}\succ_{f} T\cup \{w'\} \mbox{ if and only if } w\succ_{f}w'.$$
\end{enumerate}
\end{definition} 
 Let $\mu \succ_{F}\mu'$ denote that all firms like $\mu$ at least as well as $\mu'$ 
with at least one firm strictly preferring $\mu$ to $\mu'$, that is, $\mu(f)\succeq_{f}\mu'(f)$ for all  $f\in F$
and $\mu(f')\succ_{f'}\mu'(f')$ for at least one firm $f'\in F$. 
We say that $\mu\succeq_{F}\mu'$ means that either $\mu\succ_{F}\mu'$ or $\mu=\mu'$. Analogously, define $\mu \succ_{W}\mu'$ and $\mu \succeq_{W}\mu'$. 

We say that a matching $\mu$ is \textit{individually rational} if $\mu(w)=f$ for some worker $w$ and firm $f$, then $(f,w)$ is an acceptable pair.
Similarly, a pair $(f,w)$ is a \textit{blocking pair} for matching $\mu$, if the worker $w$ is not employed by the firm $f$, but they both prefer to be matched to each other. That is, a matching $\mu$ is blocked by a \textit{firm-worker pair} $(f,w)$:
\begin{enumerate}
\item [1.] If $\vert\mu(f)\vert=q_{f},\mu(w)\not=f ,~ f\succ_{w}\mu(w)\mbox{ and }w\succ_{f}w'$ for some $w'\in{\mu(f)}$.
\item [2.] If $\vert\mu(f)\vert<q_{f},\mu(w)\not=f \mbox{ and } f\succ_{w}\mu(w)\mbox{ and } w\succ_{f}\emptyset.$
\end{enumerate}

In that way, a matching $\mu$ is \textit{stable} if it is individually rational and has no blocking pairs. We denote by $S(P)$ to the set of all stable matchings at the preference profile $P$.

An importat result of matching theory is the \textit{Rural Hospital Theorem} (RHT). When firms have $q$-responsive preference and workers strict preference, the Rural Hospital Theorem states the following: (see Roth \cite{roth1984stability},\cite{roth1986allocation} for more details)
\medskip

\noindent \textbf{Theorem (RHT)}\textit{ The set of matched agents is the same under every stable matching. Moreover, each firm that does not fill its quota has the same set of agents matched under every stable matching.}


\subsection{Linear Programming Approach.}

For the marriage market, Rothblum \cite{rothblum1992characterization} characterizes stable matchings as extreme points of a convex polytope generated by a system of linear inequalities. Ba\"{i}ou and Balinski \cite{baiou2000stable} present two generalizations of the convex polytope for the many-to-one matching market $(F,W,P,\boldsymbol{q})$ with $\boldsymbol{q}$-responsive preferences. We focus on the most general one.

Given a matching $\mu,$ a vector $x^{\mu}\in\left\{  0,1\right\}
^{\left\vert F\right\vert \times\left\vert W\right\vert }$ is an
\textbf{incidence vector} when $x_{f,w}^{\mu}=1$ if and only if
$\mu\left(  w\right)  =f$ and $x_{f,w}^{\mu}=0$ otherwise. When no
confusion arises, we identify each matching with its incidence
vector.

 Let \textit{CP} be the convex polytope generated by the following linear inequalities:

\begin{eqnarray}
\sum_{j\in{W}}x_{f,j}\leq{q_{f}}~~~~~~~~~~~~~~~~~~~~~~~~~~~~~~~~~~~~ & f\in{F}\label{ecuacion 1} ~~~~~~~~ ~~~~~~~\\
\sum_{i\in{F}}x_{i,w}\leq{1} ~~~~~~~~~~~~~~~~~~~~~~~~~~~~~~~~~~~~~~& w\in{W}\label{ecuacion 2} ~~~~~~~~~~~~~~\\
x_{f,w}\geq{0}~~~~~~~~~~~~~~~~~~~~~~~~~~~~~~~~~~~~~~~~~~~& (f,w)\in{F\times{W}}\label{ecuacion 4}~~~~~\\
x_{f,w}={0}~~~~~~~~~~~~~~~~~~~~~~~~~~~~~~~~~~~~~~~~~~~& (f,w)\in{F\times{W}}\setminus{A}\label{ecuacion 5}
\end{eqnarray}

Notice that an integer solution of \textit{CP} represents the incidence vector of a individually rational matching for the many-to-one matching market.  The extreme points of this convex polytope  are all integer points. This convex polytope is known as the polytope of the transportation problem. For more detail, see Luenberger and Ye \cite{luenberger2008linear}. A  non-integer solution of \textit{CP} is called a \textbf{fractional matching}.

Define a new convex polytope \textit{SCP}, by adding to the convex polytope \textit{CP} the following inequality:
\begin{eqnarray}
\sum_{j\succ_{f}w}x_{f,j}+q_{f}\sum_{i\succ_{w}f}x_{i,w}+q_{f}x_{f,w}\geq q_{f} ~~~~~& (f,w)\in{A}\label{ecuacion 3}
\end{eqnarray}

\begin{lemma}[Ba\"{i}ou and Balinski \cite{baiou2000stable}] \label{Ba-Ba stable PC}
Let $(F,W,P,\boldsymbol{q})$ be a many-to-one matching market. $\mu$ is a stable matching for $(F,W,P,\boldsymbol{q})$ if and only if its incidence vector is an integer solution of \textit{SCP}.
\end{lemma}
 
 We define a \textbf{stable fractional matching} as a not necessarily integer solution of the convex polytope \textit{SCP}. 
For the marriage market, i.e. $q_{f}=1$ for all $f\in F$, Rothblum \cite{rothblum1992characterization} proves that the extreme points of the associated convex polytope, are the stable matchings. It is naturally expected that this result carries over to the more general case, a many-to-one matching market. But this is not true for the convex polytope \textit{SCP}. Here, we present an example taken from Ba\"{i}ou and Balinski \cite{baiou2000stable} that shows a many-to-one market, where the convex polytope has fractional extreme points. This also shows that a lottery over stable matchings is also a stable fractional matching. However, the opposite case does not always hold. 
\begin{example}\label{ejemplo b-b}
Let $(F,W,P,\boldsymbol{q})$ be a many-to-one matching market.
Let $F=\{f_{1},f_{2}\}$, $W=\{w_{1},w_{2},w_{3},w_{4}\}$, $P$ is the following preference profile:
$$
\begin{array}{cc}
\succ_{f_{1}}=w_{1},w_{2},w_{3},w_{4} & ~~~~\succ_{w_{1}}=f_{2},f_{1}\\ 
\succ_{f_{2}}=w_{4},w_{3},w_{2},w_{1} & ~~~~\succ_{w_{2}}=f_{2},f_{1}\\ 
& ~~~~\succ_{w_{3}}=f_{2},f_{1}\\ 
& ~~~~\succ_{w_{3}}=f_{1},f_{2}, 
\end{array}
$$
and $q_{1}=q_{2}=2$. The only two stable matchings for this market are:

\[
\begin{array}{cccc}
x^{\mu_{F}}= &\left[\begin{array}{cccc}1&1&0&0\\0&0&1&1\end{array}\right];~~~~~& x^{\mu_{W}}= &\left[\begin{array}{cccc}1&0&0&1\\0&1&1&0\end{array}\right].
\end{array}
\]
Ba\"{i}ou and Balinski observe that the stable fractional matching 
\[
\begin{array}{cc}
x^{1}= &\left[\begin{array}{cccc}1&\frac{1}{2}&\frac{1}{2}&0\\0&\frac{1}{2}&\frac{1}{2}&1\end{array}\right],
\end{array}
\]
is a vertex of the convex polytope \textit{SCP}.
\end{example}\bigskip

After observing that the convex polytope has fractional extreme points, Ba\"{i}ou and Balinski \cite{baiou2000stable} present a second generalization for the many-to-one matching market. In this second generalization, the extreme points of the convex polytope, are exactly the stable matchings for the many-to-one market. This assures that this last convex polytope,  is a subset of the convex polytope \textit{SCP}. For that reason, our study is based on the convex polytope \textit{SCP}.

\section{The Strongly Stable Fractional Matchings.}

Each entry of the vector that represents a stable fractional matching, $x_{f,w}$, can be interpreted as the time that firm $f$ and worker $w$  spend with each other. For a stable fractional matching $x$,  it can happen that two agents, one from each side of the market, have an incentive to increase the time that they spend together at the expense of those
they like less at a stable fractional matching $x$. The importance of a strongly stable fractional matching is to avoid this and prevent that agents have incentives to block the stable fractional matching in a fractional way. We formally present the definition of a strongly stable fractional matching for our market.
\bigskip
\begin{definition}
Let $(F,W,P,\boldsymbol{q})$ be a many-to-one matching market. A fractional matching $\bar{x}$ is \textbf{strongly stable}  if, for each $(f,w)\in{A}$, $\bar{x}$ satisfies the \textit{strong stability condition}  
\begin{equation}\label{condition SS(P)}
\left[  q_{f}-\sum_{j\succeq_{f}w}\bar{x}_{f,j}\right]  \cdot\left[  1-\sum_{i\succeq_{w}f}\bar{x}_{i,w}\right]  =0.   
\end{equation}
The matching $\bar{x}$ satisfying the strong stability condition is known as a \textbf{strongly stable fractional matching}.
\end{definition} 
\bigskip
We denote by $SSF(P)$ to the set of all strongly stable fractional matchings at the preference profile $P$. Assume that for a pair $(f,w) \in A$, the fractional matching $\bar{x}$ does not fulfil condition (\ref{condition SS(P)}). Then, we have that  $ q_{f}-\sum_{j\succeq_{f}w}\bar{x}_{f,j}>0$  and $ 1-\sum_{i\succeq_{w}f}\bar{x}_{i,w}>0$. Meaning that there are at least two agents $f'$ and $w'$ such that, $f\succ_{w}f'$, $w\succ_{f}w'$, $\bar{x}_{f,w'}>0$, $\bar{x}_{f',w}>0$ and  $\bar{x}_{f,w}<1$. Hence, both $f$ and $w$ will  have an incentive to increase the time that they spend together at the expense of $w'$ and $f'$ respectively. This means that $\bar{x}$ is blocked in a fractional way by the pair $(f,w)$.
\medskip

\noindent \textbf{Example 1 (Continued)} \textit{Recall the stable fractional matching $$
\begin{array}{cc}
x^{1}= &\left[\begin{array}{cccc}1&\frac{1}{2}&\frac{1}{2}&0\\0&\frac{1}{2}&\frac{1}{2}&1\end{array}\right].
\end{array}
$$
Also, we have that $f_2\succ_{w_3}f_1$  and $w_4 \succ_{f_2}w_3 \succ_{f_2}w_2$. Now, we compute condition (\ref{condition SS(P)}) for the pair $(f_2,w_3)$.
$$\left[q_{f_2}-\sum_{j\succeq_{f_2}w_3}\bar{x}_{f_2,j}\right]\cdot\left[1-\sum_{i\succeq_{w_3}f_2}\bar{x}_{i,w_3}\right]$$
$$=\left[2-\frac{3}{2}\right]\cdot\left[1-\frac{1}{2}\right]\neq 0.$$
Hence, $x^{1}$ does not fulfils condition (\ref{condition SS(P)}) for the pair $(f_2,w_3)$. Moreover, since $x^{1}_{f_1,w_3}=\frac{1}{2}>0$, $x^{1}_{f_2,w_2}=\frac{1}{2}>0$ and  $x^{1}_{f_2,w_3}=\frac{1}{2}<1$, then agents $f_2$ and $w_3$ have incentive to increase the time that they spend together at expense of $w_2$ and $f_1$ respectively. Hence, $x^{1}$ is blocked in a fractional way by the pair $(f_2,w_3)$. Therefore, $x^{1}$ is not strongly stable.}

\begin{remark}\label{matching estable cumple condicion}
The incidence vector of a stable matching, also fulfils condition (\ref{condition SS(P)}).
\end{remark}
 
For the particular case where all quotas are equal to one, (the marriage market), and for a stable fractional matching $x$, Rothblum \cite{rothblum1992characterization} defines a stable matching that assigns to each firm $f$ the most preferred worker among those that $x_{f,w}>0$, for all $w\in{W}$. Here we generalize this definition for the many-to-one matching market $(F,W,P,\boldsymbol{q})$. We denote $supp(x)$ to the support of the fractional matching $x$, that is, $supp(x)=\{(f,w):x_{f,w}>0\}$.

For a many-to-one matching market $(F,W,P,\boldsymbol{q})$, and for a given stable fractional matching $x$, we define the set of workers employed in the best $q_{f}$ positions of $f$. 
Let $C^{0}_{f}(x)=\{w:(f,w)\in{supp(x)}\}$, 
and define $C^{k}_{f}(x)=\{w\in{C_{f}^{0}(x)}:\mbox{there is no } ~w'\in{C^{0}_{f}(x)\setminus{C^{k-1}_{f}(x)},~w'\succ_{f}w}\}$. In words, $C^{k}_{f}(x)$ is the set of the $k$-best workers in the $supp(x)$ for the firm $f$. 

Now, we define the matching where each firm is matched to the best $q_{f}$ workers in the $supp(x)$. Formally,
\bigskip
\begin{definition}\label{defino mux}
Let $(F,W,P,\boldsymbol{q})$ be a many-to-one matching market. Let $x$ be a stable fractional matching. For each firm $f$, we define $\mu_{x}$ as:
$$\mu_{x}(f)=C^{q_{f}}_{f}(x).$$
\end{definition}\bigskip
\begin{remark}\label{observ mux}
If for some firm $f$ we have that $\vert C^{0}_{f}(x)\vert\leq q_{f}$, then $x^{\mu_{x}}_{f,w}=1$ for all $w\in{C^{0}_{f}(x)}$.
\end{remark}\bigskip
The following lemma generalizes  Lemma 12 of Roth \textit{et al.} \cite{roth1993stable}, and states that $\mu_{\bar{x}}$ is a stable matching whenever $\bar{x}$ is a strongly stable fractional matching.
\begin{lemma}\label{mux stable}
Let $(F,W,P,\boldsymbol{q})$ be a many-to-one matching market. Let $\bar{x}$ be a strongly stable fractional matching. Then, $\mu_{\bar{x}}$ is a stable matching.
\end{lemma}
\begin{proof}
See the Appendix.\end{proof}

The following lemma is a technical result used further in Theorem \ref{sucesion decreciente}. This lemma states that a strongly stable fractional matching $\bar{x}$ is always represented as a convex combination between the stable matching $\mu_{\bar{x}}$ and another strongly stable fractional matching. 
\begin{lemma}\label{para sus decre}
Let $(F,W,P,\boldsymbol{q})$ be a many-to-one matching market. Let $\bar{x}$ be a strongly stable fractional matching and $\bar{x}\not=x^{\mu_{\bar{x}}}$. Let $\alpha=min\{\bar{x}_{f,w}:x^{\mu_{\bar{x}}}_{f,w}=1\}$. Then, $y$ defined as:
$$
y=\frac{\bar{x}-\alpha x^{\mu_{\bar{x}}}}{1-\alpha} 
$$
is a strongly stable fractional matching, such that $supp(y)\subset supp(\bar{x})$.\footnote{Notice that, here we use ``$\subset$" to denote the strict inclusion; that is, $A\subset B$ means that $A$ is a proper subset of $B$.}

\end{lemma}
\begin{proof}See the Appendix.\end{proof}

The following theorem states that a strongly stable fractional matching can be represented by a particular convex combination of stable matchings. These stable matchings are all comparable in the eyes of all firms.
\begin{theorem}\label{sucesion decreciente}
Let $(F,W,P,\boldsymbol{q})$ be a many-to-one matching market. Let $\bar{x}$ be a strongly stable fractional matching. Then, there are stable matchings $\mu^1,\ldots,\mu^k$, and real numbers $\alpha_1,\ldots,\alpha_k$ such that 
\begin{equation}
\bar{x}=\sum_{l=1}^{k}\alpha _{l}x^{\mu ^{l}},~~ 0<\alpha _{l}\leq
1,~~ \sum_{l=1}^{k}\alpha _{l}=1,\text{ and }\mu ^{1}\succ_{F}\mu ^{2}\succ_{F}\ldots\succ_{F}\mu ^{k}.  \label{Sum-k}
\end{equation}

\end{theorem}
\begin{proof}
Let $(F,W,P,\boldsymbol{q})$ be a many-to-one matching market and let $\bar{x}$ be a strongly stable fractional matching. By Lemma \ref{mux stable}, $\mu_{\bar{x}}$ is a stable matching. Denote by  $\mu ^{1}=\mu _{\bar{x}}$.

If $\bar{x}=x^{\mu ^{1}}$ (i.e., $\bar{x}$ is a stable matching),  then $\bar{x}$ is represented as in (\ref{Sum-k}) with $k=1$ and $\alpha _{1}=1$.

If $\bar{x}$ $\neq x^{\mu ^{1}}$ (i.e., $\bar{x}$ is not a stable matching), then by Lemma \ref{para sus decre}, there is a strongly stable fractional matching, $x^{2}$,
defined by%
\[
x^{2}=\frac{\bar{x}-\alpha _{1}^{\prime }x^{\mu^{1}}}{1-\alpha _{1}^{\prime
}},
\]%
for some $0<\alpha _{1}^{\prime }<1,$ with $supp\left( x^{2}\right) \subset
supp\left( \bar{x}\right)$.  Then,
\begin{equation}
\bar{x}=\left( 1-\alpha _{1}^{\prime }\right) x^{2}+\alpha _{1}^{\prime
}x^{\mu _{1}}.  \label{Sum-1}
\end{equation}%
for $0<\alpha _{1}^{\prime }<1$ and $supp\left( x^{\mu_{1}}\right) \subset supp\left( \bar{x}\right)$.

By Lemma \ref{mux stable}, $\mu_{x^2}$ is a stable matching. Denote by $\mu ^{2}=$ $\mu _{x^{2}}$.
Notice that, since $supp\left( x^{\mu_{1}}\right) \subset {supp\left( \bar{x}\right)},~ supp\left( x^{2}\right) \subset supp\left( \bar{x}\right)$, then by definitions of $\mu^{1}$ and $x^{2}$, we have that $\mu ^{1}\succ_{F}\mu ^{2}$.

If $x^{2}=x^{\mu ^{2}}$ (i.e., $x^{2}$ is a stable matching), then $\bar{x}$ is represented as in (\ref{Sum-k}).

If $x^{2}\neq $ $x^{\mu^{2}}$ (i.e., $x^{2}$ is not a stable matching), again by Lemma \ref{para sus decre}, there is a strongly stable fractional matching $x^{3}$, defined by
\[
x^{3}=\frac{x^{2}-\alpha _{2}^{\prime }x^{\mu ^{2}}}{1-\alpha _{2}^{\prime }}%
,
\]%
for some $0<\alpha _{2}^{\prime }<1$ with $supp\left( x^{3}\right) \subset
supp\left( x^{2}\right) .$ That is,
\begin{equation}
x^{2}=\left( 1-\alpha _{2}^{\prime }\right) x^{3}+\alpha _{2}^{\prime
}x^{\mu ^{2}}.  \label{Sum-2}
\end{equation}

Since $0<\alpha _{2}^{\prime }<1$, we have that $supp\left( x^{\mu_{2}}\right)\subset{supp\left( x^{2}\right)}$. By Lemma \ref{mux stable}, $\mu_{x^3}$ is a stable matching. Denote by $\mu ^{3}=$ $\mu _{x^{3}}$
Since $supp\left(x^{3}\right) \subset supp\left( x^{2}\right) $, we have that $\mu
^{2}\succ_{F}\mu ^{3}.$ Then, $\mu ^{1}\succ_{F}\mu^{2}\succ_{F}\mu^{3}.$

If $x^{3}=x^{\mu^{3}}$ (i.e., $x^{3}$ is a stable matching), from equalities (\ref{Sum-1}) and (\ref{Sum-2}) we have that
\begin{eqnarray*}
\bar{x} &=&\left( 1-\alpha _{1}^{\prime }\right) x^{2}+\alpha _{1}^{\prime
}x^{\mu _{1}} \\
&=&\left( 1-\alpha _{1}^{\prime }\right) \left( \left( 1-\alpha _{2}^{\prime
}\right) x^{3}+\alpha _{2}^{\prime }x^{\mu ^{2}}\right) +\alpha _{1}^{\prime
}x^{\mu ^{1}} \\
&=&\left( 1-\alpha _{1}^{\prime }\right) \left( 1-\alpha _{2}^{\prime
}\right) x^{3}+\left( 1-\alpha _{1}^{\prime }\right) \alpha _{2}^{\prime
}x^{\mu ^{2}}+\alpha _{1}^{\prime }x^{\mu ^{1}}.
\end{eqnarray*}%
Then $\bar{x}$ is represented as in (\ref{Sum-k}) with $k=3$, $\alpha _{1}=\alpha _{1}^{\prime
}, $ $\alpha _{2}=\left( 1-\alpha _{1}^{\prime }\right) \alpha _{2}^{\prime
},$ and $\alpha _{3}=\left( 1-\alpha _{1}^{\prime }\right) \left( 1-\alpha
_{2}^{\prime }\right)$. Notice that $\alpha_{1}+\alpha_{2}+\alpha_{3}=1$.

If $x^{3}\neq $ $x^{\mu _{3}}$ (i.e., $x^{3}$ is not a stable matching), then we continue this procedure. The finiteness of the $\mathit{supp}\left(  \bar{x}\right)  $
guarantees that this procedure ends by constructing a stable matching.
This proves that $\bar{x}$ is represented as in (\ref{Sum-k}) for some $k\geq 1$.\end{proof}

Recall that the convex polytope $SCP$ has extreme points that are not integer. A salient question now is, are these non-integer extreme points strongly stable fractional matchings?
The following corollary answers this, and states that non-integer extreme points of the convex polytope $SCP$ are not strongly stable fractional matchings.
\begin{corollary}\label{vertice fraccionario no fuertemente estable}
Let $(F,W,P,\boldsymbol{q})$ be a many-to-one matching market. Let $x$ be a non-integer extreme point of the convex polytope \textit{SCP}. Then, $x$ is not a strongly stable fractional matching.
\end{corollary}  
\begin{proof}
Let $x$ be a non-integer extreme point of the convex polytope $SCP$. Then, $x$ cannot be represented as a convex combination of different extreme points of the same convex polytope. More precisely, $x$ cannot be represented as a convex combination of different \emph{integer} extreme points of the convex polytope $SCP$ (stable matchings). Therefore, by Theorem \ref{sucesion decreciente}, $x$ is not a strongly stable fractional matching.\end{proof}

\subsection{Cycles in Preferences.}

For the marriage market, Irving and Leather \cite{irving1986complexity} define a cycle in preference and a cyclic
matching in order to present an algorithm that finds all stable matchings. Bansal \textit{et al.} \cite{bansal2007polynomial}
 and Cheng \textit{et al.} \cite{cheng2008unified} extend the concept of cycles and cyclic matchings for many-to-many and many-to-one matching markets, respectively. We will state some properties of cycles that are taken from these authors. They refer to the cycles as rotations.
  
Given a stable matching $\mu$ for a many-to-one matching market $(F,W,P,\boldsymbol{q})$, we define a \textbf{reduced preference profile} $\boldsymbol{P^{\mu}}$, as the preference profile obtained after the reduction procedure. This reduction procedure is presented in the Appendix.
The reduced preference list of firm $f$, is denoted by $\succ_{f}^{\mu}$ . In the same way, the reduced preference list of worker $w$, is denoted by  $\succ_{w}^{\mu}$ .
\bigskip
\begin{definition}
Let $(F,W,P,\boldsymbol{q})$ be a many-to-one matching market. Given a stable matching $\mu$, and the reduced preference profile $P^{\mu}$, a set of firms $\sigma=\{e_{1},\ldots,e_{r}\}\subseteq{F}$ defines a \textbf{cycle} if for $w_{e_{1}},\ldots,w_{e_{r}}\in W$ we have that:
\begin{enumerate}
\item[1.] For each $d=1,\ldots,r-1$, $w_{e_{d}}\in{\mu(e_{d+1})}$, $w_{e_{d}}\not\in{\mu(e_{d})}$ and  $w_{e_{d}}\succeq_{e_{d}}w'$ for all $w'\not\in{\mu(e_{d})}$.
\item[2.] $w_{e_{r}}\not\in{\mu(e_{r})}$, $w_{e_{r}}\succeq_{e_{r}}w'$ for all $w'\not\in{\mu(e_{r})}$, and  $w_{e_{r}}\in{\mu(e_{1})}.$
\end{enumerate}
\end{definition}
\bigskip
Given a cycle $\sigma$, we can define a cyclic matching as follows:

\begin{definition}
Let $(F,W,P,\boldsymbol{q})$ be a many-to-one matching market. Given a stable matching $\mu$, and the reduced preference profile $P^{\mu}$, let $\sigma=\{e_{1},\ldots,e_{r}\}$ be a cycle in $P^{\mu}$, and let $\{w_{e_{1}},\ldots,w_{e_{r}}\}$ be the set of workers defined by the cycle $\sigma$. The \textbf{cyclic matching} of $\mu$ is defined as follows:
$$
\mu[\sigma]=\left\{
\begin{array}{lll}
\mu[\sigma](e_{1})=&\mu(e_{1})\setminus\{w_{e_{r}}\}\cup\{w_{e_{1}}\},&\\
\mu[\sigma](e_{d})=&\mu(e_{d})\setminus\{w_{e_{d-1}}\}\cup\{w_{e_{d}}\}&\mbox{ for $d=2,\ldots,r-1$},\\
\mu[\sigma](e_{r})=&\mu(e_{r})\setminus\{w_{e_{r-1}}\}\cup\{w_{e_{r}}\}&\\
\mu[\sigma](f)=&\mu(f)&\mbox{ for all $f\not\in{\sigma}$}.
\end{array}\right.
$$
\end{definition}

In the reduced preference lists of each firm that belongs to the cycle $\sigma$, we have that the preferred worker that is unmatched under $\mu$ is always matched to another firm in the same cycle. Think of each firm $e_d$ in the cycle as being asked to hire its preferred worker that is unmatched in its preference list. Also, this new worker replaces the worker that other firm in the cycle wants to hire.  The firms that do not belong to the cycle $\sigma$, will be matched to the same set of workers. 
Notice that if a firm $f$ belongs to a cycle $\sigma$, this means that it has different sets of workers assigned in $\mu$ as well as in $\mu[\sigma]$. Then, by Theorem \textcolor{blue}{RHT}, we have that $\vert\mu(f)\vert=q_{f}.$
\begin{lemma}[Bansal \textit{et al.} \cite{bansal2007polynomial}]\label{ciclico es estable}
Let $(F,W,P,\boldsymbol{q})$ be a many-to-one matching market. 
\begin{enumerate}
\item[1.] Let $\mu$ be a stable matching and let $\sigma$ be a cycle in $P^{\mu}$. Then, the cyclic matching $\mu[\sigma]$ is a stable matching in the original preference profile.
\item[2.] A matching $\mu'$ is stable under $P^{\mu}$ if and only if $\mu'$ is stable under the original preference profile and $\mu\succeq_{F}\mu'$.
\end{enumerate}
\end{lemma}

Let $\boldsymbol{\Phi(\mu)}$ denote the set of cycles of the reduced preference profile $P^{\mu}.$
Now, we can extend the definition of a cyclic matching as follows.\bigskip
\begin{definition}\label{defino matching ciclico generico}
Let $(F,W,P,\boldsymbol{q})$ be a many-to-one matching market. For a stable matching $\mu$, and the reduced preference profile $P^{\mu}$, let $K\subseteq\Phi(\mu)$, define the \textbf{cyclic matching} $\boldsymbol{\mu[K]}$ as follows:
\begin{enumerate}
\item If $K=\emptyset$, then $\mu[K]=\mu.$
\item If $K\neq\emptyset$, and $K=\{\sigma_{1},\ldots,\sigma_{n}\}$, then
$$
\mu[K](f)=\left\{
\begin{array}{ll}
\mu[\sigma_{h}](f)&~~~~f\in{\sigma_{h}},~ h=1,\ldots,n\\
\mu(f) & ~~~~\mbox{otherwise.}
\end{array}
\right.
$$

\end{enumerate} 
\end{definition}

\begin{lemma}[Cheng \textit{et al.} \cite{cheng2008unified}] \label{ciclico disjuntos}
Let $(F,W,P,\boldsymbol{q})$ be a many-to-one matching market, and let  $P^{\mu}$ be the reduced preference profile at $\mu$.
\begin{enumerate}
\item[1.] Let $\sigma$ and $\sigma'$ be two different cycles. Then, $\sigma\cap\sigma'=\emptyset$.
\item[2.] Let $\mu'$ be a stable matching in $P^{\mu}$. If $\mu\neq\mu'$, then there is a cycle $\sigma\in{P^{\mu}}$  such that $\mu[\sigma]\succeq_{F}\mu'.$
\end{enumerate}
\end{lemma}
\begin{remark}\label{mu de K es igual a mu de sigma}
Let $K\subseteq\Phi(\mu)$ be a subset of cycles of $P^{\mu}$. By Lemma (\ref{ciclico disjuntos}), we have that $\mu[K](f)=\mu[\sigma](f)$ for each $f\in{\sigma}$ with $\sigma\in{K}$ .
\end{remark}

Notice that Lemmas \ref{ciclico es estable} and \ref{ciclico disjuntos} assure that the cyclic matching $\mu[K]$ from Definition \ref{defino matching ciclico generico} is stable under the original preference profile.

The following lemma states that, the matching obtained by applying different cycles is independent from the order in which they are applied.
\begin{lemma}\label{independencia del orden}
Let $(F,W,P,\boldsymbol{q})$ be a many-to-one matching market. Let $P^{\mu}$ be the reduced preference profile at $\mu$, and let $\sigma$ and $\sigma'$ be two different cycles  in $\Phi(\mu)$.

\begin{enumerate}[1.]
\item $\sigma$ is a cycle of $P^{\mu[\sigma']}$.
\item $\mu[\sigma,\sigma']=\mu[\sigma',\sigma].$
\end{enumerate}
\end{lemma}
\begin{proof}
See the Appendix.\end{proof}

\section{A characterization of the set of strongly stable fractional matchings.}
In this section, we present our main findings. Our aim is to describe the relationship among the stable matchings involved in the lottery that represents a strongly stable fractional matching.   
Our characterization of a strongly stable fractional matching is based on the idea of cyclic matchings.  For a many-to-one matching market, finding all stable matchings via cycles in preference and cyclic matchings requires a polynomial time algorithm, (see Bansal \textit{et al.} \cite{bansal2007polynomial}). The importance of our characterization is to present an elegant and useful way to describe the strongly stable fractional matchings via the stable matchings involved in the lottery that generates them.  
Since our result is based the idea of cycles and cyclic matchings,  we need that a stable fractional matching that is strongly stable in a reduced preference profile is also strongly stable in the original preference profile. This statement is proved on Lemma \ref{ss en reducido entonces ss en original} in the Appendix.

In order to present our characterization, we define a connected set generated by a stable matching.\bigskip
\begin{definition}\label{defino conjunto conectado}
Let $(F,W,P,\boldsymbol{q})$ be a many-to-one matching market. A set of stable matchings $\mathcal{M}$ is \textbf{connected} if there is a stable matching $\mu$ and a set of cycles $K'\subseteq\Phi(\mu)$ such that
$$
\mathcal{M}=\mathcal{M}_{\mu}^{K'},
$$
where $\mathcal{M}_{\mu}^{K'}=\{\mu[K]:K\subseteq K'\}.$ Let us denote by $\mathcal{M}_\mu=\mathcal{M}_\mu^{\Phi(\mu)}.$
\end{definition}

Notice that from the Definition \ref{defino matching ciclico generico}, we can see that $\mu$ is also a cyclic matching of itself and, for each $K'\subseteq\Phi(\mu)$, we have that  $\mu\in{\mathcal{M}^{K'}_{\mu}}$.

The main result of this paper states that $\bar{x}$ is a strongly stable fractional matching if and only if it belongs to the convex hull of a connected set. Formally:
\begin{theorem}\label{caraterizacion}
Let $(F,W,P,\boldsymbol{q})$ be a many-to-one matching market and let $\bar{x}$ be a stable fractional matching in $( F,W,P).$ Then, $\bar{x}$ is strongly stable if and only if there is a collection of connected stable matchings $\{\mu ^{1},\ldots,\mu ^{k} \}$, such that $\bar{x}=\sum_{l=1}^{k}\alpha _{l}x^{\mu ^{l}},~ 0<\alpha _{l}\leq
1,\text{ and } \sum_{l=1}^{k}\alpha _{l}=1.$
\end{theorem}
\begin{proof}
Let $(F,W,P,\boldsymbol{q}) $ be a many-to-one matching market.

 $\boldsymbol{(\Longrightarrow)}$ Let $\bar{x}$ be a strongly stable fractional matching. Theorem \ref{sucesion decreciente},  assures that there are stable matchings $\mu^1,\ldots,\mu^k$ and real numbers $\alpha_1,\ldots,\alpha_k$ such that 
$$
\bar{x}=\sum_{l=1}^{k}\alpha _{l}x^{\mu ^{l}},~~ 0<\alpha _{l}\leq
1,~~ \sum_{l=1}^{k}\alpha _{l}=1,  \text{ and }\mu ^{1}\succ_{F}\mu ^{2}\succ_{F}\ldots\succ_{F}\mu ^{k}.
$$

Denote by $Conv\left\{ \mathcal{M}_{\mu ^1}\right\}$ the convex hull of elements of $ \mathcal{M}_{\mu ^1}.$ 
Assume, by way of contradiction, that  $\bar{x}\notin Conv\left\{ \mathcal{M}_{\mu ^1}\right\}$. Then,  there is $\mu^{t}$ of the convex combination of $\bar{x}$ such that $\mu^1,\ldots,\mu^{t-1}\in\mathcal{M}_{\mu ^1}$ and  $\mu^{t}\notin \mathcal{M}_{\mu ^1}$.
Let $K'\subseteq\Phi(\mu)$ be the set of cycles such that $\mu[K']=\mu^{t-1}$. Notice that, $\mathcal{M}^{K'}_{\mu ^1}\subseteq \mathcal{M}_{\mu ^1}$ is the smallest connected set generated by $\mu^1$ such that $\mu^1,\ldots,\mu^{t-1}\in\mathcal{M}^{K'}_{\mu ^1}$ and  $\mu^{t}\notin \mathcal{M}^{K'}_{\mu ^1}$.

Then, by Lemma \ref{ciclico es estable} item 2) and Lemma \ref{ciclico disjuntos} item 2)  there is a cycle $\sigma^{\star}\in\Phi\left(  \mu^{t-1}\right)  $
 such that $\mu^{t-1}[\sigma^{\star}]\not \in \mathcal{M}^{K'}_{\mu^{1}}$  and $\mu^{t-1}[ \sigma^{\star}] \succeq_{F}\mu^{t}.$ Notice that this implies that $\sigma^{\star}\notin K'$.

Here we analyse two cases:
\begin{itemize}
\item[\textbf{Case 1:}]\textbf{ If there is $\boldsymbol{\sigma\in K'}$ such that $\boldsymbol{\sigma^{\star}\cap\sigma\neq \emptyset}.$} 

Notice that in this case, by Lemma \ref{ciclico disjuntos} item 1), $\sigma^{\star}\notin \Phi(\mu^1)$. Then for any $\tilde{f}\in{\sigma^{\star}\cap\sigma}$, we have that, $\mu^{1}\succ_{\tilde{f}}\mu^{t-1}\succ_{\tilde{f}}
\mu^{t-1}[\sigma^{\star}]  \succeq_{\tilde{f}}\mu^{t}$.
If $t=2$, we have that $\sigma^{\star}\in{K'}$, and since $\sigma\cap \sigma^{\star}\neq\emptyset$, then $\sigma=\sigma^{\star}$ which results in a contradiction. Therefore $t\geq3$.

By Theorem \textcolor{blue}{RHT}, there are $w^{\star},~w_{1},~w_{2}$ such that:
\begin{equation}\label{w en el medio}
\begin{array}{l}
w_{1}\in{\mu^{1}(\tilde{f})-\left(\mu^{t-1}(\tilde{f})\cup \mu^{t-1}[\sigma^{\star}](\tilde{f})\right)},\\
w^{\star}\in{\mu^{t-1}(\tilde{f})-\mu^{1}(\tilde{f})},\\
w_{2}\in{\mu^{t-1}[\sigma^{\star}](\tilde{f})-\left(\mu^{t-1}(\tilde{f})\cup \mu^{1}(\tilde{f})\right)},
\end{array}
\end{equation}

and $w_{1}\succ_{\tilde{f}}w^{\star}\succ_{\tilde{f}}w_{2}$. Now, we prove that for the pair $(\tilde{f},w^{\star})$, condition (\ref{condition SS(P)}) fails. That is, 
$$
\left[q_{\tilde{f}}-\sum_{j\succeq_{\tilde{f}}w^{\star}}\bar{x}_{\tilde{f},j}\right]\cdot\left[1-\sum_{i\succeq_{w^{\star}}\tilde{f}}\bar{x}_{i,w^{\star}}\right]\neq 0.
$$
We analyse the two factors separately:
\begin{itemize}
\item[\textbf{Case 1.1:}] $\boldsymbol{q_{\tilde{f}}-\sum_{j\succeq_{\tilde{f}}w^{\star}}\bar{x}_{\tilde{f},j}.}$
$$
q_{\tilde{f}}-\sum_{j\succeq_{\tilde{f}}w^{\star}}\bar{x}_{\tilde{f},j}=q_{\tilde{f}}-\sum_{j\succeq_{\tilde{f}}w^{\star}}\left(\sum_{l=1}^{k}\alpha_{l}x^{\mu^{l}}_{\tilde{f},j}\right)=
$$
$$
q_{\tilde{f}}-\sum_{l=1}^{k}\alpha_{l}\left(\sum_{j\succeq_{\tilde{f}}w^{\star}}x^{\mu^{l}}_{\tilde{f},j}\right).
$$
Since $ \mu^{t} \preceq_{\tilde{f}} \mu^{t-1}[\sigma^{\star}] $, we have that $$\sum_{j\succeq_{\tilde{f}}w}x^{\mu^{t}}_{\tilde{f},j}\leq \sum_{j\succeq_{\tilde{f}}w}x^{\mu[K'][\sigma^{\star}]}_{\tilde{f},j}$$ for all $w\in{W}$. 

In particular for  $w=w^{\star}$, and the fact that $w^{\star}\succ_{\tilde{f}}w^{2}$ with $w^{2}\in{\mu^{t-1}[\sigma^{\star}](\tilde{f})}$, we have that $$\sum_{j\succeq_{\tilde{f}}w^{\star}}x^{\mu^{t}}_{\tilde{f},j}\leq\sum_{j\succeq_{\tilde{f}}w^{\star}}x^{\mu^{t-1}[\sigma^{\star}]}_{\tilde{f},j} <q_{\tilde{f}}.$$
By (\ref{w en el medio}) we also have that for each $l=1,\ldots,k$
$$
\sum_{j\succeq_{\tilde{f}}w^{\star}}x^{\mu^{l}}_{\tilde{f},j}\leq q_{\tilde{f}}.
$$
Using the decreasing sequence of stable matchings of Theorem \ref{sucesion decreciente}, we have that $\alpha_{l}>0$ for each $l=1,\ldots,k$. Then, we have that 
$$
q_{\tilde{f}}-\sum_{j\succeq_{\tilde{f}}w^{\star}}\bar{x}_{\tilde{f},j}=q_{\tilde{f}}-\sum_{l=1}^{k}\alpha_{l}\left(\sum_{j\succeq_{\tilde{f}}w^{\star}}x^{\mu^{l}}_{\tilde{f},j}\right)>q_{\tilde{f}}-\left(\sum_{l=1}^{k}\alpha_{l}q_{\tilde{f}}\right)=0.
$$
\item[\textbf{Case 1.2:}]$\boldsymbol{1-\sum_{i\succeq_{w^{\star}}\tilde{f}}\bar{x}_{i,w^{\star}}.}$
$$
1-\sum_{i\succeq_{w^{\star}}\tilde{f}}\bar{x}_{i,w^{\star}}=1-\sum_{i\succeq_{w^{\star}}\tilde{f}}\left(\alpha_{1}x^{\mu^{1}}_{i,w^{\star}}+\sum_{l=2}^{k}\alpha_{l}x^{\mu^{l}}_{i,w^{\star}}\right)=
$$
$$\left[1-\left(\alpha_{1}\sum_{i\succeq_{w^{\star}}\tilde{f}}x^{\mu^{1}}_{i,w^{\star}}+\sum_{l=2}^{k}\alpha_{l}\sum_{i\succeq_{w^{\star}}\tilde{f}}x^{\mu^{l}}_{i,w^{\star}}\right)\right].
$$
By (\ref{w en el medio}), we have that $w^{\star}\not\in{\mu^{1}(\tilde{f})}$. Since $\mu^{1}\succ_{\tilde{f}}\mu^{t-1}$, we have that $\mu^{t-1}(w^{\star})=\tilde{f}\succ_{w^{\star}}\mu^{1}(w^{\star})$. Therefore, 
$$
\sum_{i\succeq_{w^{\star}}\tilde{f}}x^{\mu^{1}}_{i,w^{\star}}=0.
$$
Since $\alpha_l>0$ for each $l=1,\ldots,k$, then
$$1-\left(\alpha_{1}\sum_{i\succeq_{w^{\star}}\tilde{f}}x^{\mu^{1}}_{i,w^{\star}} +\sum_{l=2}^{t}\alpha_{l}\sum_{i\succeq_{w^{\star}}\tilde{f}}x^{\mu^{l}}_{i,w^{\star}}\right)=
$$ 
$$1-\sum_{l=2}^{t}\alpha_{l}\sum_{i\succeq_{w^{\star}}\tilde{f}}x^{\mu^{l}}_{i,w^{\star}}=1-\sum_{l=2}^{t}\alpha_{l}>0.
$$ 
\end{itemize}
Then from cases 1.1 and 1.2, we have that for the pair $(\tilde{f},w^{\star})$, 
$$
\left[q_{\tilde{f}}-\sum_{j\succeq_{\tilde{f}}w^{\star}}\bar{x}_{\tilde{f},j}\right]\cdot\left[1-\sum_{i\succeq_{w^{\star}}\tilde{f}}\bar{x}_{i,w^{\star}}\right]\neq 0.
$$
 That is, for the pair $(\tilde{f},w^{\star})$ condition (\ref{condition SS(P)}) fails.

\item[\textbf{Case 2:}]\textbf{If $\boldsymbol{\sigma\cap\sigma^{\star}=\emptyset}$ for all $\boldsymbol{\sigma\in{K'}}$.}

Notice that in this case, $\sigma^{\star}$ may or may not belong to $\Phi(\mu^1).$
Notice also that $\mu^{t-1}(f)\neq\mu^{t-1}[\sigma^{\star}](f)$ for all $f\in{\sigma^{\star}}$.
\begin{itemize}
\item[\textbf{Case 2.1:}] $\boldsymbol{\sigma\notin{\Phi(\mu^1)}}$.

We claim that there are $\bar{f}\in{\sigma^{\star}}$ and $\bar{w}\in W\setminus \{\mu^{1}(\bar{f)}) \cup \mu^{t-1}(\bar{f})\}$, such that for 
$\mu^{1}(\bar{f})=\{w_{1}^{\mu^{1}},\ldots,w_{q_{\bar{f}}}^{\mu^{1}}\}$,  $\mu^{t-1}(\bar{f})=\{w_{1}^{\mu^{t-1}},\ldots,w_{q_{\bar{f}}}^{\mu^{t-1}}\}$, we have that $w_{1}^{\mu^{1}}\succ_{\bar{f}}\bar{w}\succ_{\bar{f}}w_{q_{\bar{f}}}^{\mu^{t-1}}$.

If not, for all $f\in{\sigma^{\star}}$, we have that
$
\mu^{1}(f)\succ_{f}\mu^{t-1}[\sigma^{\star}](f)\succ_{f}w,
$
for each $w\not\in\{\mu^{1}(f) \cup \mu^{t-1}[\sigma^{\star}](f)\}$.\footnote{Here, $\mu^{t-1}[\sigma^{\star}](f)\succ_{f}w$ denotes that worker $w$ is less preferred for the firm $f$ than all workers matched to firm $f$ under the stable matching $\mu^{t-1}[\sigma^{\star}]$.}
Since $\sigma\cap\sigma^{\star}=\emptyset$, we have that $\mu^{1}(f)=\mu^{t-1}(f)$ for all $f\in{\sigma^{\star}}$. Then, let $\{w^{f}\}=\mu^{t-1}[\sigma^{\star}](f)\setminus \mu^{1}(f)$ for each $f\in{\sigma^{\star}}$. That is, $w^{f}$ is the most preferred worker in the reduced preference list $P^{\mu^{1}}(f)$ such that it does not belong to $\mu^{1}(f)$. This implies that $\sigma^{\star}\in\Phi(\mu^{1})$, and it is a contradiction since $\sigma^{\star}\not\in\Phi(\mu^{1})$.

Therefore, there are $\bar{f}\in{\sigma^{\star}}$ and $\bar{w}\in W\setminus \{\mu^{1}(\bar{f)}) \cup \mu^{t-1}(\bar{f})\}$, such that  
\begin{equation}\label{desigualdad 2 teorema FE}
\mu^{1}(\bar{f})=\mu^{t-1}(\bar{f})\succ_{\bar{f}}\bar{w}\succ_{\bar{f}}\mu^{t-1}[\sigma^{\star}](\bar{f})
\end{equation}
Since $\sigma^{\star}\in{\Phi(\mu^{t-1})}$, in order to obtain the reduced preference lists $P^{\mu^{t-1}}$, $\bar{f}$ should have eliminated $\bar{w}$ by means of the third step of the reduction procedure.
Then, we have that 
\begin{equation}\label{desigualdad 3 teorema FE}
\mu^{t}(\bar{w})\succeq_{\bar{w}}\mu^{t-1}[\sigma^{\star}](\bar{w})\succ_{\bar{w}}\mu^{t-1}(\bar{w})\succ_{\bar{w}}\bar{f}\succ_{\bar{w}}\mu^{1}(\bar{w}).
\end{equation}
Since $\bar{x}$ can be written as in Theorem \ref{sucesion decreciente}, ($\mu^{1},\mu^{t-1}\in{\mathcal{M}^{K'}_{\mu^{1}}}$ and  $\mu^{t-1}\succ_{F}\mu^{t}$), and using inequalities (\ref{desigualdad 2 teorema FE}) and (\ref{desigualdad 3 teorema FE}), we have that
$$
\sum_{j\succeq_{\bar{f}}\bar{w}}\bar{x}_{\bar{f},j}<q_{\bar{f}}\mbox{   and   }\sum_{i\succeq_{\bar{w}}\bar{f}}\bar{x}_{i,\bar{w}}<1.
$$
Then, 

$$
\left[q_{\tilde{f}}-\sum_{j\succeq_{\bar{f}}\bar{w}}\bar{x}_{\bar{f},j}\right]\cdot\left[1-\sum_{i\succeq_{\bar{w}}\bar{f}}\bar{x}_{i,\bar{w}}\right]> 0.
$$
That is, condition (\ref{condition SS(P)}) fails for the pair $(\bar{f},\bar{w})$.
\item[\textbf{Case 2.2:}] $\boldsymbol{\sigma\in{\Phi(\mu^1)}}$.

In this case, we have that $\mu^{t-1}[\sigma^{\star}]\in \mathcal{M}_{\mu^1}\setminus \mathcal{M}^{K'}_{\mu^1}.$ Then by Lemma \ref{ciclico es estable} item 2) and Lemma \ref{ciclico disjuntos} item 2), there is a cycle $\sigma'\in\Phi\left(  \mu^{t-1}[\sigma^{\star}]\right)  $ such that $$\mu^{t-1}\succ_F\mu^{t-1}[\sigma^{\star}]\succ_F\mu^{t-1}[\sigma^{\star}][\sigma']\succ_F\mu^{t}.$$
Notice that, $\sigma'$ may or may not belong to $\Phi(\mu^1)$. If $\sigma \notin \Phi(\mu^1)$, the arguments follows as in Case 2.1. If $\sigma' \in \Phi(\mu^1)$ we continue this process until, by finiteness of the set $\Phi(\mu^1)$, there is $\tilde{\sigma}\notin \Phi(\mu^1)$, and the arguments follows as in Case 2.1.
\end{itemize}

\end{itemize}
Therefore, from cases 1 and 2, we have that there is a connected set $\mathcal{M}_{\mu^{1}}$ such that  $\bar{x}\in{Conv\left( \mathcal{M}_{\mu^{1}}\right)}.$

$\boldsymbol{(\Longleftarrow)}$ Let $\bar{x}$ be
a convex combination of stable matchings from a connected set. That is, there are a stable matching $\mu$, and a list of sets $K_{1},\ldots,K_{k}\subseteq\Phi(\mu)$ with the corresponding cyclic matchings $\mu^{1},\ldots,\mu^{k}$, such that
$
\bar{x}=\sum_{l=1}^{k}\alpha_{l}x^{\mu^{l}}
$ with  $0\leq\alpha_{l}\leq1$ and
$\sum _{l=1}^{k}\alpha_{l}=1$.

Since $\bar{x}$ is a convex combination of stable matchings from $\mathcal{M}_{\mu}$, we have that $\mu \succeq_{F} \mu^{l}$ for each $l=1,\ldots,k.$ Then, we have that $\bar{x}$ is a stable fractional matching  for the matching market $(F,W,P^{\mu})$. Moreover, since $S(P^{\mu})\subseteq S(P)$, we have that $\bar{x}$ is also a stable fractional matching  for the matching market $(F,W,P,\boldsymbol{q})$.
By Lemma \ref{ss en reducido entonces ss en original}, we only need to prove that $\bar{x}$ is strongly stable in the reduced preference profile $P^{\mu}$.

If $\alpha_{1}=1$ we have that $\bar{x}=x^{\mu^{1}}$. Since $\mu^{1}$ is also a stable matching in the original preferences profile, then we have that $\bar{x}$ is strongly stable. Hence, we assume $0<\alpha_{l}<1$ for each $l=1,\ldots,k.$ 
Now, we prove that $\bar{x}$ fulfils condition (\ref{condition SS(P)}) for each pair $(f,w)\in A(P^{\mu})$.

Fix $f\in F$. 
Assume that firm $f$ does not fill its quota. Theorem \textcolor{blue}{RHT} assures that this firm is always assigned to the same set of workers in every stable matching. Then $\bar{x}_{f,j}=x_{f,j}^{\mu^{l}}$ for $l=1,\ldots,k$ and for all $j$ such that $(f,j)\in A(P^{\mu})$. Since $\mu^{l}$ is a stable matching for $l=1,\ldots,k$, it fulfils condition (\ref{condition SS(P)}) for each $(f,j)\in A(P^{\mu})$. Then, by Remark \ref{matching estable cumple condicion} we have that $\bar{x}$ also fulfil condition (\ref{condition SS(P)}) for each $(f,j)\in A(P^{\mu})$.

Assume now that $f$ does fill its quota. 
Let $\mathcal{K=}%
{\textstyle\bigcup\nolimits_{l=1}^{k}}
K_{l}$. Let $\mu(f)=\{w_{1},\ldots,w_{q_{f}}\}$ and $w_{i}\succ_{f}w_{i+1}$. We analyse two cases separately.
\begin{enumerate}
\item[\textbf{Case 1:}] \textbf{There is no $\boldsymbol{\sigma\in\mathcal{K}}$ such that $\boldsymbol{f\in\sigma.}$} 

By Definition \ref{defino matching ciclico generico},we have that $\mu[\mathcal{K}](f)=\mu(f)$. That is, for each $j\in{W}$, 
$\bar{x}_{f,j}=x_{f,j}^{\mu}$. 

Thus, if $w\preceq_{f}w_{q_{f}}$ we have that
$$
\sum_{j\succeq_{f}^{\mu}w}\bar{x}_{f,j}=\sum_{j\succeq_{f}^{\mu}w}x_{f,j}^{\mu}=q_{f}.
$$

If $w\succ_{f}w_{q_{f}}$, we have that  
$$
\sum_{j\succeq_{f}^{\mu}w}\bar{x}_{f,j}=\sum_{j\succeq_{f}^{\mu}w}x_{f,j}^{\mu}< q_{f}.
$$

Then, $\sum_{j\succ_{f}^{\mu}w}x_{f,j}^{\mu}< q_{f}.$ Moreover, by linear inequality (\ref{ecuacion 3}) for the stable matching $\mu$, we have that $\sum_{i\succeq_{w}^{\mu}f}x_{i,w}^{\mu}>0.$ Therefore, $\sum_{i\succeq_{w}^{\mu}f}x_{i,w}^{\mu}=1.$ Recall that $\mu$ is the firm oltimal stable matching in the reduced preference profile $P^{\mu}$, then Lemma \ref{fuertemente al medio} states that $\bar{x}\succeq^{\mu}_{W}x^{\mu}$, i.e., $$\sum_{i\succeq_{w}^{\mu}f}\bar{x}_{i,w}\geq\sum_{i\succeq_{w}^{\mu}f}x_{i,w}^{\mu}=1.$$ 

By the linear inequality (\ref{ecuacion 2}), we have that $\sum_{i\succeq_{w}^{\mu}f}\bar{x}_{i,w}=1$. Thus, for $(f,w)\in A(P^{\mu})$, we have that 
$$
\left[  q_{f}-\sum_{j\succeq_{f}^{\mu}w}\bar{x}_{f,j}\right]  \cdot\left[  1-\sum_{i\succeq_{w}^{\mu}f}\bar{x}_{i,w}\right]=0.
$$

\item[\textbf{Case 2:}] \textbf{There is $\boldsymbol{\sigma_{f}\in\mathcal{K}}$ such that $\boldsymbol{f\in\sigma_{f}.}$}

 By Lemma
\ref{ciclico disjuntos}, there is a unique cycle
$\sigma_{f}\in K_{l}$ such that $f\in\sigma_{f}.$ But $\sigma_{f}$
may be in more than one set $K_{l}.$ We denote $L_{f}=\{l:\sigma_{f}\in{K_{l}}\}$. Therefore,
$$
\bar{x}_{f,j}=\sum_{l=1}^{k}\alpha_{l}x_{f,j}^{\mu^{l}}=\sum_{l\in L_{f}}\alpha_{l}x_{f,j}^{\mu^{l}}+\sum_{l\not\in L_{f}}\alpha
_{l}x_{f,j}^{\mu^{l}}.
$$

Since $\sigma_{f}$ is unique, by Lemma \ref{ciclico disjuntos} and Lemma \ref{independencia del orden}, we have that $\mu[K_{l}](f)=\mu[\sigma _{f}](f)$ and
$x_{f,j}^{\mu^{l}}=x_{f,j}^{\mu[\sigma_{f}]}$ for those ${l\in L_{f}}$. Also, $\mu[K_{l}](f)=\mu(f)$ and $x_{f,j}^{\mu^{l}}=x_{f,j}^{\mu}$ for those ${l\not\in L_{f}}$. Hence,
$$
\sum_{l\in L_{f}}{\alpha_{l}x_{f,j}^{\mu^{l}}}+\sum_{l\not\in L_{f}}{\alpha_{l}x_{f,j}^{\mu^{l}}}=\sum_{l\in L_{f}}{\alpha_{l}x_{f,j}^{\mu[\sigma_{f}]}}+\sum_{l\not\in L_{f}}{\alpha_{l}x_{f,j}^{\mu}}=
$$
\begin{equation}\label{equ A}
x_{f,j}^{\mu[\sigma_{f}]}\left(\sum_{l\in L_{f}}\alpha_{l}\right)+x_{f,j}^{\mu}\left(\sum_{l\not\in L_{f}}\alpha_{l}\right).
\end{equation}

Since $\sum_{l\in L_{f}}\alpha_{l}  + \sum_{l\not\in L_{f}}\alpha_{l}=1$, then we define
$\bar{\alpha}=\sum_{l\in L_{f}}\alpha_{l}$.

Then (\ref{equ A}) is equal to
$\bar{\alpha}x_{f,j}^{\mu[\sigma_{f}]}+\left(  1-\bar{\alpha}\right)  x_{f,j}^{\mu}$. That is, $\bar{x}_{f,j}$ is the convex combination of $x_{f,j}^{\mu[\sigma_{f}]}$ and $x_{f,j}^{\mu}$. Since $f\in\sigma_{f}$, then  
$$
\vert supp(x_{f,\cdot}^{\mu[\sigma_{f}]})\vert = \vert supp(x_{f,\cdot}^{\mu})\vert =q_{f} \mbox{ and } \vert supp(x_{f,\cdot}^{\mu[\sigma_{f}]})\cap supp(x_{f,\cdot}^{\mu}) \vert = q_{f}-1. 
$$

Hence, there are two workers $w_{a}$ and $w_{q_{f}+1}$ such that $(f,w_{a})\in{supp(x^{\mu}) \setminus supp(x^{\mu[\sigma_{f}]})}$ and  $(f,w_{q_{f}+1})\in{supp(x^{\mu[\sigma_{f}]}) \setminus supp(x^{\mu})}$. Let $T=\{w_{s}:(f,w_{s})\in{supp(x^{\mu}) \cap supp(x^{\mu[\sigma_{f}]})}\}$. Then,

\begin{equation}\label{para el corolario}
\bar{x}_{f,j}=\left\{
\begin{array}{ll}
1 &~\mbox{if }j\in{T} \\
1-\bar{\alpha} &~\mbox{if } j=w_{a}\\
\bar{\alpha} &~\mbox{if } j=w_{q_{f}+1}.
\end{array}
\right.
\end{equation}

Now, we prove that $\bar{x}$ fulfils condition (\ref{condition SS(P)}) in $P^{\mu}$ for each $w$: 
\begin{enumerate}
\item[i)]
If $w\preceq_{f}w_{q_{f}+1}$, then
$$\sum_{j\succeq_{f}^{\mu}w}\bar{x}_{f,j}=\sum_{s\in{T}}\bar{x}_{f,w_{s}}+\bar{x}_{f,w_{a}}+\bar{x}_{f,w_{q_{f}+1}}=
(q_{f}-1)+(1-\bar{\alpha})+\bar{\alpha}=q_{f}.$$
Then, 
$$
 q_{f}-\sum _{j\succeq_{f}^{\mu}w}\bar{x}_{f,j}=0,
$$
hence,
$$
\left[  q_{f}-\sum _{j\succeq_{f}^{\mu}w}\bar{x}_{f,j}\right]  \cdot\left[  1-\sum_{i\succeq_{w}^{\mu}f}\bar{x}_{i,w}\right]=0.
$$
\item[ii)]
If  $w\succ^{\mu}_{f}w_{q_{f}}$, then
$$
\sum _{j\succeq_{f}^{\mu}w}\bar{x}_{f,j}<q_{f}.
$$
Since 
$$
q_{f}>\sum _{j\succeq_{f}^{\mu}w}\bar{x}_{f,j}=(1-\bar{\alpha})\sum _{j\succeq_{f}^{\mu}w}x^{\mu}_{f,j}+\bar{\alpha}\sum _{j\succeq_{f}^{\mu}w}x^{\mu[\sigma_{f}]}_{f,j},
$$
then, 
$$
\sum _{j\succeq_{f}^{\mu}w}x^{\mu}_{f,j}<q_{f}\mbox{ and }\sum _{j\succeq_{f}^{\mu}w}x^{\mu[\sigma_{f}]}_{f,j}<q_{f}.
$$
But $\mu$ and $\mu[\sigma]$ are stable matchings, and these stable matchings fulfil condition (\ref{condition SS(P)}). So we have that 
$$
\sum_{i\succeq_{w}^{\mu}f}x^{\mu}_{i,w}=1\mbox{ and } \sum_{i\succeq_{w}^{\mu}f}x^{\mu[\sigma_{f}]}_{i,w}=1,
$$
in which case we can assure that 
$$
\sum_{i\succeq_{w}^{\mu}f}\bar{x}_{i,w}=(1-\bar{\alpha})\sum_{i\succeq_{w}^{\mu}f}x^{\mu}_{i,w}+\bar{\alpha}\sum_{i\succeq_{w}^{\mu}f}x^{\mu[\sigma_{f}]}_{i,w}=1.
$$
That is, 
$$
\left[  q_{f}-\sum _{j\succeq_{f}^{\mu}w}\bar{x}_{f,j}\right]  \cdot\left[  1-\sum_{i\succeq_{w}^{\mu}f}\bar{x}_{i,w}\right]=0.
$$
\item[iii)]
If $w=w_{q_{f}}$. Recall that $\mu(f)=\{w_{1},\ldots,w_{q_{f}}\}$, then $\mu(w)=f$. Also, we have that , $\mu\succ_{F}^{\mu}\mu^{l}$ and $\mu^{l}\succeq_{W}^{\mu}\mu$ for all $l=1,\ldots,k$. In particular, $\mu^{l}(w)\succeq_{w}^{\mu} \mu(w)=f$ for all $l=1,\ldots,k$. This implies that 
$$
\sum_{i\succeq_{w}^{\mu}f}x_{i,w}^{\mu^{l}}=1
$$
for all $l=1,\ldots,k.$ 
Hence
$$
\sum_{i\succeq_{w}^{\mu}f}\bar{x}_{i,w}=\sum_{i\succeq_{w}^{\mu}f}\sum_{l=1}^{t}\alpha_{l}x_{i,w}^{\mu^{l}}=\sum_{l=1}^{t}\alpha_{l}\sum_{i\succeq_{w}^{\mu}f}x_{i,w}^{\mu^{l}}=\sum_{l=1}^{t}\alpha_{l}=1.
$$

Then, 
$$
\left[  q_{f}-\sum _{j\succeq_{f}^{\mu}w}\bar{x}_{f,j}\right]  \cdot\left[  1-\sum_{i\succeq_{w}^{\mu}f}\bar{x}_{i,w}\right]=0.
$$

\end{enumerate} 
\end{enumerate}
From cases 1 and 2, we have that for the pair $(f,w)$,  
$$
\left[  q_{f}-\sum _{j\succeq_{f}^{\mu}w}\bar{x}_{f,j}\right]  \cdot\left[  1-\sum_{i\succeq_{w}^{\mu}f}\bar{x}_{i,w}\right]=0.
$$

Therefore, $\bar{x}$ is a strongly stable fractional matching in the reduced preference profile $P^{\mu}$, and by Lemma \ref{ss en reducido entonces ss en original}, $\bar{x}$ is a strongly stable fractional matching in the original preference profile $P$.
\end{proof}

Once we have characterized all strongly stable fractional matchings for the many-to-one matching market $(F,W,P,\boldsymbol{q})$, we can characterize the set of all strongly stable fractional matchings. Recall that, $SSF(P)$ denotes the set of all strongly stable fractional matchings at the preference profile $P$.
\begin{corollary}\label{corolario 2}
Let $(F,W,P,\boldsymbol{q})$ be a many-to-one matching market. Then,
$$
SSF(P)=\bigcup_{\mu\in{S(P)}}Conv\{\mathcal{M}_{\mu}\}.
$$
\end{corollary}

The following corollary extends Corollary 21 from Roth \textit{et al.} \cite{roth1993stable}. It gives an upper bound to the number of worker matched to each firm, and the number of firms matched to each worker.

\begin{corollary}\label{corolario 3}
Let $(F,W,P,\boldsymbol{q})$ be a many-to-one matching market. Each strongly stable fractional matching fulfils the following two conditions:
\begin{enumerate}
\item Each worker has a positive probability with at most two distinct firms.
\item Each firm, all but possibly one position are assigned deterministically. For the one position that is assigned by a lottery, two workers have a positive probability of been employed.
\end{enumerate}
\end{corollary}

Notice that, from (\ref{para el corolario}) in the proof of Theorem \ref{caraterizacion}, the previous corollary follows straightforward. Another extension is due to  Schlegel \cite{schlegel2018note} for a school choice matching market with strict preferences (similar setting than ours). Our characterization gives an alternative proof for these two similar results, for the school choice set-up due to Schlegel \cite{schlegel2018note} is straightforward, and for the marriage market due to Roth \textit{et al.} \cite{roth1993stable}, it's necessary only to set all quotas of all firms equal to one. 

\section*{Conclusions.}
In this paper we present a strong stability condition for a many-to-one matching market where firms' preference are $q$-responsive. Further, we prove that a strongly stable fractional matching can be represented as a convex combination of stable matchings that fulfil a decreasing order in the eyes of all firms. Although it was already known the ``almost" integrability of a strongly stable fractional matchings, there may be more ``almost" integral stable fractional matchings that are not strongly stable (the stable fractional matching $x^1$ in Example \ref{ejemplo b-b} illustrates this). Our characterization of strongly stable fractional matching allow us to describe precisely which are the agents matched.  Also, we characterize the set of all strongly stable fractional matchings. We think that our results gives a complete description of the strongly stable fractional matchings.

\section*{Acknowledgements.}
We would like to thank Agustin Bonifacio, Jordi Mass\'o and the Game Theory Group
of IMASL for the helpful discussions and  detailed comments.
Also, we are grateful with ``MISS CLAIR" for the assistance
with the English writing. Our work is partially supported by \textit{Universidad Nacional de San Luis}, through grant 319502, and by the
\textit{Consejo Nacional de Investigaciones Cient\'{\i}ficas y T\'{e}cnicas}
(CONICET), through grant PIP 112-200801-00655.

\appendix
\section{Appendix}

\subsection*{The reduction procedure:}

Let $\left( F,W,P,\boldsymbol{q}\right)$ be a many-to-one matching market. Let $\mu_{F}$ be the optimal stable matching for all firms, and $\mu_{W}$ be the optimal stable matching for all workers. 

\noindent\textbf{Step 1: }Remove all $w$ who are more preferred than the most preferred worker matched under
$\mu _{F}(f)  $ from $f$'s list of acceptable workers. Remove all $f$ who are more preferred than $\mu_{W}(w)$ from $w$'s list of acceptable firms.

Therefore, the most preferred worker matched in $\mu_{F}(f)$ will be the first entry in
$f$'s reduced list, and $\mu_{W}(w)  $ will be
the first entry in $w$'s reduced list.

\noindent\textbf{Step 2: }Remove all $f$ who are less preferred than
$\mu _{F}\left(  w\right)  $ from $w$'s list of acceptable
firms. Remove all $w$ who are less preferred than the least worker matched under $\mu_{W}\left(
f\right)  $ from $f$'s list of acceptable workers.

Thus, $\mu_{F}\left(  w\right)  $ will be the last entry in
$w$'s reduced list and the least preferred  worker in $\mu_{W}\left(  f\right)  $ will be
the last entry in $f$'s reduced list.

\noindent\textbf{Step 3: }After steps 1 and 2, if $f$ is not acceptable for $w$ (i.e., if
$f$ is not on $w$'s preference list as now modified), then
remove $w$ from $f$'s list of acceptable workers, and similarly, remove from  $w$'s list of acceptable firms, any firm $f$ to whom $w$ is no longer acceptable.

Hence, $f$ will be acceptable for $w$ if and only if $w$ is
acceptable for $f$ after Step 3.

For the matching market $(F,W,P^{\mu},\boldsymbol{q})$, the stable matching $\mu $ is the  $%
F-optimal$ stable matching, that is the stable matching that all firm prefer in the matching market $(F,W,P^{\mu},\boldsymbol{q})$.

\subsection*{Lemmas and Proofs.}

\begin{proof}[Proof of Lemma \ref{mux stable}]
Let $(F,W,P,\boldsymbol{q})$ be a many-to-one matching market. Let $\bar{x}$ be a strongly stable fractional matching. First, we will prove that $\mu_{\bar{x}}$ is a matching.
Assume that all  positive entries of $\bar{x}$ are equal to $1$, then we have by Definition \ref{defino mux} that $x^{\mu_{\bar{x}}}=\bar{x}$. Since  $\bar{x}$ is a strongly stable fractional matching, by Lemma \ref{Ba-Ba stable PC} we have that $\mu_{\bar{x}}$ is a stable matching.

Assume now that not all  positive entries of $\bar{x}$ are equal to $1$. 
We will prove that $\mu_{\bar{x}}$ is a matching. Assume that is not a matching. That is, there is a worker $w$  and two different firms $f$ and $f'$, such that $w\in{\mu_{\bar{x}}(f)}$ and $w\in{\mu_{\bar{x}}(f')}$. Since the preferences of the worker $w$ are strict, without loss of generality, we can assume that $f\succ_{w}f'$.
We will show that $\sum_{i\succeq_{w}f}\bar{x}_{i,w}=1$, for this we analyse two cases:
\begin{enumerate}
\item[\textbf{Case 1:}] $\boldsymbol{\vert C_{f}^{0}(\bar{x}) \vert \leq q_{f}}$.

 We have that if $\sum_{j\succeq_{f}w}\bar{x}_{f,j}<q_{f}$, since $\bar{x}$ is a stable fractional matching. Then condition (\ref{condition SS(P)}) implies that
$$
\sum_{i\succeq_{w}f}\bar{x}_{i,w}=1.
$$
Hence, $\sum_{i\prec_{w}f}\bar{x}_{i,w}=0$, and $\bar{x}_{f',w}=0$, which contradicts the assumption of $\bar{x}_{f',w}>0$.

If $\sum_{j\succeq_{f}w}\bar{x}_{f,j}=q_{f}$, then $w\in{C_{f}^{q_{f}}}$ and $\bar{x}_{f,w}=1$. This implies that
$$
\sum_{i\succeq_{w}f}\bar{x}_{i,w}=1.
$$
\item[\textbf{Case 2:}]$\boldsymbol{\vert C_{f}^{0}(\bar{x})\vert>q_{f}.}$

Notice that $C^{q_{f}}_{f}(\bar{x})\subset C^{0}_{f}(\bar{x})$. Since $w\in{\mu_{\bar{x}}(f)}$, then $w\in{C^{q_{f}}_{f}(\bar{x})}$, and we have that
$$
\sum_{j\succeq_{f}w}\bar{x}_{f,j}\leq\sum_{j\in{C^{q_{f}}_{f}(\bar{x})}}\bar{x}_{f,j}<\sum_{j\in{C^{0}_{f}(\bar{x})}}\bar{x}_{f,j}\leq q_{f}
$$
Hence, $\sum_{j\succeq_{f}w}\bar{x}_{f,j}< q_{f}$. Since $\bar{x}$ is a strongly stable fractional matching, condition (\ref{condition SS(P)}) implies that
$$
\sum_{i\succeq_{w}f}\bar{x}_{i,w}=1.
$$
\end{enumerate}
From cases 1 and 2, we have that $\sum_{i\succeq_{w}f}\bar{x}_{i,w}=1$, then $\sum_{i\prec_{w}f}\bar{x}_{i,w}=0$, and also $\bar{x}_{f',w}=0$ since $f\succ_{w}f'$, which contradicts the assumption of $\bar{x}_{f',w}>0$. Therefore $\mu_{\bar{x}}$ is a matching.

Now, we will prove that  $\mu_{\bar{x}}$ is a stable matching. Let $w\in{\mu_{\bar{x}}(f)}$, then $w\in{C_{f}^{q_{f}}(\bar{x})\subseteq C_{f}^{0}(\bar{x})}$. Hence $w\succ_{f}f$ and $f\succ_{w}w$. Then $\mu_{\bar{x}}$ in an individually rational matching.

Assume that there is a blocking pair $(\bar{f},\bar{w})$ of $\mu_{\bar{x}}$. This means that we have the following three statements:
\begin{enumerate}
\item[a)]$\bar{w}\not\in{\mu_{\bar{x}}(\bar{f})}$.
\item[b)]There is $w'\in{\mu_{\bar{x}}(\bar{f})}$ such that: either $\bar{w}\succ_{\bar{f}}w'$ if $\vert\mu_{\bar{x}}(\bar{f})\vert=q_{\bar{f}}$,  or $\bar{w}\succ_{\bar{f}}\bar{f}$ if $\vert\mu_{\bar{x}}(\bar{f})\vert<q_{\bar{f}}$.
\item[c)]$\bar{f}\succ_{\bar{w}}\mu_{\bar{x}}(\bar{w})$.
\end{enumerate}
Now, we will show that $\bar{x}_{\bar{f},\bar{w}}=0$:
\begin{enumerate}
\item[i)]If $\vert C^{q_{f}}_{f}(\bar{x}) \vert < q_{f}$, then $C^{q_{f}}_{f}(\bar{x})=C^{0}_{f}(\bar{x})$. Hence, since  $\mu_{\bar{x}}(\bar{f})=C^{0}_{f}(\bar{x})$ and $\bar{w}\not\in{\mu_{\bar{x}}(\bar{f})}$, we have that $\bar{x}_{\bar{f},\bar{w}}=0$.

\item[ii)]If $\vert C^{q_{\bar{f}}}_{\bar{f}}(\bar{x}) \vert=q_{\bar{f}}$. Let $w^{\star}\in  C^{q_{\bar{f}}}_{\bar{f}}(\bar{x}) \setminus  C^{q_{\bar{f}}-1}_{\bar{f}}(\bar{x}).$ That is, $w^{\star}$ is the least preferred worker employed with firm $\bar{f}$ under the matching  $\mu_{\bar{x}}.$ Then, by item b), we have that
\begin{equation}\label{eq prueba mux}
 \sum_{j\succeq_{\bar{f}}\bar{w}}\bar{x}_{\bar{f},j}<\sum_{j\succeq_{\bar{f}}w^{\star}}\bar{x}_{\bar{f},j}\leq q_{\bar{f}}.
\end{equation}

 By items a) and b), we have that $\bar{w}\succ_{\bar{f}}w^{\star},$ and by Definition \ref{defino mux} we can assure that $\bar{x}_{\bar{f},\bar{w}}=0$. If not, we have that $\bar{w}\in{\mu_{\bar{x}}(\bar{f})},$ a contradiction.
\end{enumerate}

By (\ref{eq prueba mux}) and the fact that $\bar{x}$ by hypothesis is strongly stable, we have that 
$$
1-\sum_{i\succeq_{\bar{w}}\bar{f}}\bar{x}_{i,\bar{w}}=0.
$$
Then, 
$$
1=\sum_{i\succeq_{\bar{w}}\bar{f}}\bar{x}_{i,\bar{w}}=\bar{x}_{\bar{f},\bar{w}}+\sum_{i\succ_{\bar{w}}\bar{f}}\bar{x}_{i,\bar{w}}=0+\sum_{i\succ_{\bar{w}}\bar{f}}\bar{x}_{i,\bar{w}}.
$$
Hence, $\sum_{i\prec_{\bar{w}}\bar{f}}\bar{x}_{i,\bar{w}}=0$, but this is a contradiction since from Definition \ref{defino mux}, we have that $\bar{x}_{\mu_{\bar{x}}(\bar{w}),\bar{w}}>0$ and by item c) we have $\bar{f}\succ_{\bar{w}}\mu_{\bar{x}}(\bar{w})$. Therefore, $\mu_{\bar{x}}$ is a stable matching. \end{proof}

\begin{proof}[Proof of Lemma \ref{para sus decre}]
Let $\bar{x}$ be a strongly stable fractional matching. Then, by Lemma (\ref{mux stable}) we have that $\mu_{\bar{x}}$ is a stable matching. By Definition \ref{defino mux} we have that $supp(x^{\mu_{\bar{x}}})\subset supp(\bar{x})$, and by the fact that $\bar{x}\neq x^{\mu_{\bar{x}}}$ we have that
$$
\bar{x}=\alpha x^{\mu_{\bar{x}}} +(1-\alpha) y.
$$
We need to prove that 
$$y=\frac{\bar{x}-\alpha x^{\mu_{\bar{x}}}}{1-\alpha}$$
is a strongly stable fractional matching. That is, $y$ is a solution of \textit{CP} and fulfils condition (\ref{condition SS(P)}).

From $\bar{x}\not=x^{\mu_{\bar{x}}}$, we have that $\alpha>0$. From definition of $\alpha$, we have that $\alpha<1$.

Assume that $\alpha=\bar{x}_{\bar{f},\bar{w}}$, with $\bar{w}\in{C^{q_{\bar{f}}}_{\bar{f}}(\bar{x})}$.
\begin{itemize}
\item \textbf{Inequality (1) of \textit{CP}.} 
Following from the definition of $y$ and the definition of $\alpha$, we have that:

If $\vert C_{f}^{0}(\bar{x})\vert \geq q_{f}$, then
$$\sum_{j\in W}\bar{x}_{f,j}-\alpha\sum_{j\in W}x_{f,j}^{\mu_{\bar{x}}}=\sum_{i\in{W}}\bar{x}_{f,j}-\alpha q_{f} \leq q_{f}-\alpha q_{f}.$$
 
Therefore,
$$
\sum_{j\in W}y_{f,j}=\frac{1}{1-\alpha}\left[\sum_{j\in W}\bar{x}_{f,j}-\alpha\sum_{j\in W}x_{f,j}^{\mu_{\bar{x}}}\right]\leq q_{f}.
$$ 
If $\vert C_{f}^{0}(\bar{x})\vert =r <q_{f}$, then $\sum_{j\in W}\bar{x}_{f,j}-\alpha\sum_{j\in W}x_{f,j}^{\mu_{\bar{x}}}=\sum_{i\in{W}}\bar{x}_{f,j}-\alpha r \leq r-\alpha r=r(1-\alpha)<q_{f}(1-\alpha)$. Then,
$$
\sum_{j\in W}y_{f,j}=\frac{1}{1-\alpha}\left[\sum_{j\in W}\bar{x}_{f,j}-\alpha\sum_{j\in W}x_{f,j}^{\mu_{\bar{x}}}\right]\leq q_{f}.
$$
That is, $y$ satisfy linear inequality (1). 
\item \textbf{Inequality (2) of \textit{CP}.} 
A similar argument that is used for inequality (1), proves that $y$ satisfy linear inequality (2).
\item \textbf{Inequality (3) of \textit{CP}.}
If $(f,w)\in{supp(\bar{x})}$, by definition of $\mu_{\bar{x}}$, we shall consider two cases:
\begin{itemize}
	\item If $(f,w)\not\in{supp(x^{\mu_{\bar{x}}})}$, that is, $x^{\mu_{\bar{x}}}_{f,w}=0$. Then, $y_{f,w}=\frac{\bar{x}_{f,w}}{1-\alpha}=0$.
	\item If $(f,w)\in{supp(x^{\mu_{\bar{x}}})}$, that is, $x^{\mu_{\bar{x}}}_{f,w}=1$ and also $x^{\mu_{\bar{x}}}_{\bar{f},\bar{w}}=1$.
\end{itemize}
Then, $y_{f,w}=\frac{\bar{x}_{f,w}-\alpha x^{\mu_{\bar{x}}}_{f,w}}{1-\alpha}=\frac{\bar{x}_{f,w}-\alpha}{1-\alpha}\geq 0$. Then, for $(f,w)\in{A}$, $y_{f,w}\geq0$, that is $y$ satisfies linear inequality (3).

\item \textbf{Inequality (4) of \textit{CP}.}
Inequality (4) it is easily satisfied form definition of $y$.

\item \textbf{Condition (\ref{condition SS(P)}).} 
By hypothesis we have that $\bar{x}$ is a strongly stable fractional matching, then $\bar{x}$ fulfils condition (\ref{condition SS(P)}). Hence, 
$$\left[q_{f}-\sum_{j\succeq_{f}w}\bar{x}_{f,j}\right]\cdot\left[1-\sum_{i\succeq_{w}f}\bar{x}_{i,w}\right]=0,$$
for each $(f,w)\in A(P)$. 
Since $\bar{x}=\alpha x^{\mu_{\bar{x}}}+(1-\alpha)y$, with $0<\alpha<1$, then for each $(f,w)\in A(P)$ we have that
$$\left[q_{f}-\sum_{j\succeq_{f}w}\left(\alpha x^{\mu_{\bar{x}}}_{f,j}+(1-\alpha)y_{f,j}\right)\right]\cdot\left[1-\sum_{i\succeq_{w}f}\left(\alpha x^{\mu_{\bar{x}}}_{i,w}+(1-\alpha)y_{i,w}\right)\right]=0$$
$$\left[q_{f}-\alpha\sum_{j\succeq_{f}w} x^{\mu_{\bar{x}}}_{f,j}-(1-\alpha)\sum_{j\succeq_{f}w}y_{f,j}\right]\cdot\left[1-\alpha \sum_{i\succeq_{w}f}x^{\mu_{\bar{x}}}_{i,w}-(1-\alpha)\sum_{i\succeq_{w}f}y_{i,w}\right]=0$$
$$\left[ \alpha \left( q_{f}-\sum_{j\succeq_{f}w} x^{\mu_{\bar{x}}}_{f,j} \right)+(1-\alpha) \left( q_{f}-\sum_{j\succeq_{f}w}y_{f,j} \right) \right] \cdot$$
\begin{equation}\label{ecuacion para lema final}
\left[ \alpha \left(1- \sum_{i\succeq_{w}f} x^{\mu_{\bar{x}}}_{i,w} \right)+(1-\alpha) \left( 1-\sum_{i\succeq_{w}f} y_{i,w} \right) \right]=0.
\end{equation}

Since $x^{\mu_{\bar{x}}}$ and $y$ fulfil inequality (\ref{ecuacion 1}) and (\ref{ecuacion 2}), then
$$q_{f}-\sum_{j\succeq_{f}w}x^{\mu_{\bar{x}}}_{f,j}\geq 0,~~q_{f}-\sum_{j\succeq_{f}w}y_{f,j}\geq 0,$$
\begin{equation}\label{ecuacion para lema final 2}
sum_{i\succeq_{w}f}x^{\mu_{\bar{x}}}_{f,j}\geq 0\mbox{ and }1-\sum_{j\succeq_{w}f}y_{f,j}\geq 0.
\end{equation}

Then, by (\ref{ecuacion para lema final}) and (\ref{ecuacion para lema final 2}), we have that either
$$q_{f}-\sum_{j\succeq_{f}w}x^{\mu_{\bar{x}}}_{f,j}=0\mbox{ and } q_{f}-\sum_{j\succeq_{f}w}y_{f,j}=0$$
or
$$1-\sum_{i\succeq_{w}f}x^{\mu_{\bar{x}}}_{i,w}=0\mbox{ and } 1-\sum_{i\succeq_{w}f}y_{i,w}=0.$$

Since, by Lemma \ref{mux stable} $\mu_{\bar{x}}$ is a stable matching, and by Remark \ref{matching estable cumple condicion} $x^{\mu_{\bar{x}}}$ fulfils condition (\ref{condition SS(P)}) for each $(f,w)\in A(P)$. Therefore, $y$ fulfils condition (\ref{condition SS(P)}) for each $(f,w)\in A(P)$.

\end{itemize}
 Since $supp(x^{\mu_{\bar{x}}})\subseteq supp(\bar{x})$, we have that $supp(y)\subseteq supp(\bar{x})$.  Moreover, since $y_{\bar{f},\bar{w}}=0$, and $x^{\mu_{\bar{x}}}_{\bar{f},\bar{w}}=1$, then $supp(y)\subset supp(\bar{x})$. \end{proof}\bigskip

\begin{proof}[Proof of Lemma \ref{independencia del orden}]
Let $P^{\mu}$ be a profile of reduced lists for $\left( F,W,P\right)$, and let $\sigma$ and $\sigma'$ be two different cycles in $P^{\mu}.$ 
\begin{enumerate}[1.]
\item By Lemma \ref{ciclico disjuntos}, we have that $\sigma\cap\sigma'=\emptyset$. Then, we can assume that $\sigma'=\{e_{1},\ldots,e_{r}\}$ and $\sigma=\{e_{r+1},\ldots,e_{r+r'}\}$. By Definition \ref{defino matching ciclico generico}, 
$$\mu[\sigma'](f)=\left\{
\begin{array}{lll}
\mu(e_{1})&~~~~\mbox{if } f=e_{r}\\
\mu(e_{k+1}) & ~~~~\mbox{if }f=e_{k}, ~k=1,\ldots,r-1\\
\mu(f) & ~~~~\mbox{if }f\not\in{\{e_{1},\ldots,e_{r}\}}.
\end{array}
\right.
$$
That is, the firms that do not belong to the cycle $\sigma'$ do not change the set of workers in both stable matchings ($\mu$ and $\mu[\sigma']$). Then, by Lemma \ref{ciclico disjuntos}, the cycle $\sigma$ is a subset of those firms that do not change. That is, $\sigma$ is a cycle of $P^{\mu[\sigma']}$.

\item Let $K=\{\sigma,\sigma'\}$.
Since $K$ is a set of cycles (not an ordered set), by Definition \ref{defino matching ciclico generico} and item 1., we have that
$$\mu[K]\left(  f\right)  =\left\{
\begin{array}
[c]{ll}%
\mu[\sigma](f)  &~~~~~~~~\mbox{if }f\in{\sigma}\\
\mu[\sigma'](f)  &~~~~~~~~\mbox{if }f\in{\sigma'}\\
\mu(f) &~~~~~~~~\mbox{otherwise.}  \\
\end{array}
\right.
$$ 

Then, $\mu[\sigma',\sigma]=\mu[\sigma,\sigma']$.

\end{enumerate}
\end{proof}

We say that a fractional matching $x$ \textbf{weakly dominates} a fractional matching $y$ with respect to the preference of the firm $f$, if for all workers $w$,
$$\sum_{j\succeq_{f}w}x_{f,j} \geq \sum_{j\succeq_{f}w}y_{f,j},$$ and it will be denoted by $x\succeq_{f}y$, using the same notation that is used for stable matchings. 

Similarly  $x$ \textbf{strongly dominates} $y$, denoted by $x\succ_{f}y$, if the previous inequality holds strictly for at least one worker $w$. Weak and strong domination under a worker's preferences are defined analogously. We say that $x\succeq_{F}y$ when $x\succeq_{f}y$ for all $f\in{F}$. The relation $x\succeq_{W}y$ is defined analogously.

\begin{lemma}\label{fuertemente al medio}
Let $(F,W,P,\boldsymbol{q})$ be a many-to-one matching market. Let $\bar{x}$ be a strongly stable fractional matching. Then, $x^{\mu_{F}}\succeq_{f}\bar{x}\succeq_{f}x^{\mu_{W}}$ for all $f\in{F}$ and  $x^{\mu_{W}}\succeq_{w}\bar{x}\succeq_{w}x^{\mu_{F}}$ for all $w\in{W}$.
\end{lemma}
\begin{proof}
Let $(F,W,P,\boldsymbol{q})$ be a many-to-one matching market. Let $\bar{x}$ be a strongly stable fractional matching. By Theorem \ref{sucesion decreciente}, there are stable matchings $\mu^1,\ldots,\mu^{k}$  and real numbers $\alpha_1,\ldots,\alpha_{k}$, such that $\bar{x}=\sum_{l=1}^{k}\alpha_{l}x^{\mu^{l}}$, with $0<\alpha_{l}\leq 1$, $\sum_{l=1}^{k}\alpha_{l}=1$  and $\mu ^{1}\succ_{F}\mu ^{2}\succ_{F}\ldots\succ_{F}\mu ^{k}$.
Since  $\mu_{F}\succeq_{F}\mu^{l}\succeq_{F}\mu_{W}$ for each $l=1,\ldots,k$, then for $f\in{F}$ we have that
$$\sum_{j\succeq_{f}w}x_{f,j}^{\mu_{F}}=\sum_{j\succeq_{f}w}\left(\sum_{l=1}^{k}\alpha_{l}x_{f,j}^{\mu_{F}}\right)=\sum_{l=1}^{k}\alpha_{l}\left(\sum_{j\succeq_{f}w}x_{f,j}^{\mu_{F}}\right)\geq$$
$$\sum_{l=1}^{k}\alpha_{l}\left(\sum_{j\succeq_{f}w}x_{f,j}^{\mu^{l}}\right)=\sum_{j\succeq_{f}w}\left(\sum_{l=1}^{k}\alpha_{l}x_{f,j}^{\mu^{l}}\right)=\sum_{j\succeq_{f}w}\bar{x}_{f,j},$$ for all $w\in{W}$. Then $x^{\mu_{F}}\succeq_{f}\bar{x}$.

To prove that $\bar{x}\succeq_{f}x^{\mu_{W}}$,$$\sum_{j\succeq_{f}w}\bar{x}_{f,j}=\sum_{j\succeq_{f}w}\left(\sum_{l=1}^{k}\alpha_{l}x_{f,j}^{\mu^{l}}\right)=\sum_{l=1}^{k}\alpha_{l}\left(\sum_{j\succeq_{f}w}x_{f,j}^{\mu^{l}}\right)\geq$$ 
$$\sum_{l=1}^{k}\alpha_{l}\left(\sum_{j\succeq_{f}w}x_{f,j}^{\mu_{W}}\right)=\sum_{j\succeq_{f}w}\left(\sum_{l=1}^{k}\alpha_{l}x_{f,j}^{\mu_{W}}\right)=\sum_{j\succeq_{f}w}x_{f,j}^{\mu_{W}},$$ for all $w\in{W}$. Then $\bar{x}\succeq_{f}x^{\mu_{W}}$.

A similar argument proves that $x^{\mu_{W}}\succeq_{w}\bar{x}\succeq_{w}x^{\mu_{F}}$.\end{proof}\\

\begin{lemma}\label{ss en reducido entonces ss en original}
Let $(F,W,P,\boldsymbol{q})$ be a many-to-one matching market. Let $\mu\in{S(P)}$, and $P^{\mu}$ the reduced preference profile. Let $\bar{x}$ be a stable fractional matching for a many-to-one matching market $(F,W,P,\boldsymbol{q})$. If $\bar{x}$ is a strongly stable fractional matching for a many-to-one matching market $(F,W,P^{\mu},\boldsymbol{q})$, then $\bar{x}$ is a strongly stable fractional matching for a many-to-one matching market $(F,W,P,\boldsymbol{q})$.
\end{lemma}
\begin{proof}
Let $(F,W,P,\boldsymbol{q})$ be a many-to-one matching market. Let $\mu\in{S(P)}$, and $P^{\mu}$ be the reduced preference profile. Let $\bar{x}$ be a stable fractional matching for a many-to-one matching market $(F,W,P,\boldsymbol{q})$. Let $\bar{x}$ be a strongly stable fractional matching for a many-to-one matching market $(F,W,P^{\mu},\boldsymbol{q})$, that is:
$$
\left[  q_{f}-\sum_{j\geq^{\mu}_{f}w}\bar{x}_{f,j}\right]  \cdot\left[  1-\sum_{i\geq^{\mu}_{w}f}\bar{x}_{i,w}\right]  =0,
$$
for all $(f,w)\in{A(P^{\mu})}$.

We need to prove that, for all $(f,w)\in{A(P)}$, $\bar{x}$ fulfils
$$
\left[  q_{f}-\sum_{j\succeq_{f}w}\bar{x}_{f,j}\right]  \cdot\left[  1-\sum_{i\succeq_{w}f}\bar{x}_{i,w}\right]  =0.
$$

We consider the following two cases:  
\begin{enumerate}
\item Let $(f,w)\in{A(P^{\mu})}$. That is, $(f,w)$ was not eliminated in $P^{\mu}$. So,
$$\sum_{j\succeq_{f}w}\bar{x}_{f,j}\geq \sum_{j\geq^{\mu}_{f}w}\bar{x}_{f,j}$$
holds, since for each firm $f$, there are more workers in the original preference list than in the reduced preference list.

Hence,$$q_{f}-\sum_{j\succeq_{f}w}\bar{x}_{f,j}\leq q_{f}-\sum_{j\geq^{\mu}_{f}w}\bar{x}_{f,j}.$$

With a similar argument we have that 
$$1-\sum_{i\succeq_{w}f}\bar{x}_{i,w}\leq 1-\sum_{i\geq^{\mu}_{w}f}\bar{x}_{i,w}.$$

By hypothesis, and linear inequalities (1) and (2) of $PC$,
$$
0=\left[  q_{f}-\sum_{j\geq^{\mu}_{f}w}\bar{x}_{f,j}\right]  \cdot\left[  1-\sum_{i\geq^{\mu}_{w}f}\bar{x}_{i,w}\right]\geq 
 \left[  q_{f}-\sum_{j\succeq_{f}w}\bar{x}_{f,j}\right]  \cdot\left[  1-\sum_{i\succeq_{w}f}\bar{x}_{i,w}\right]\geq 0.
$$
Then, for $(f,w)\in{A(P^{\mu})}$, we have that 
$$\left[  q_{f}-\sum_{j\succeq_{f}w}\bar{x}_{f,j}\right]  \cdot\left[  1-\sum_{i\succeq_{w}f}\bar{x}_{i,w}\right]=0.$$

\item Let $(f,w)\in{A(P)-A(P^{\mu})}$. Let $w_1 \in \mu(f)$ such that $w_1 \succeq_f w'$ for each $w'\in \mu(f)$. Notice that $w\neq w_1$. Then, we analyse two sub-cases:

\begin{enumerate}
\item[i)]$w\succ_{f}w_{1}$.

We have that $x^{\mu}\succeq_{F}\bar{x}$, then
\begin{equation}\label{ecuacion 4.2}
\sum_{j\succ_{f}w}\bar{x}_{f,j}\leq\sum_{j\succ_{f}w}x_{f,j}^{\mu}=0.
\end{equation}
Since, $\bar{x}$ is a stable fractional matching, $\bar{x}$ satisfy inequality (\ref{ecuacion 3}) of $SPC$, i.e.
$$
\sum_{j\succ_{f}w}\bar{x}_{f,j}+q_{f}\sum_{i\succ_{w}f}\bar{x}_{i,w}+q_{f}\bar{x}_{f,w}\geq q_{f}.
$$
Then, by condition (\ref{ecuacion 4.2}) $$q_{f}\sum_{i\succeq_{w}f}\bar{x}_{i,w}\geq q_{f},$$ and for all $i\succeq_{w}f$, $\bar{x}_{i,w}=1$. Hence, $$\sum_{i\succeq_{w}f}\bar{x}_{i,w}\geq 1,$$ and by linear inequality (2), we have that $$\sum_{i\succeq_{w}f}\bar{x}_{i,w}= 1,$$ then we have that $$\left[  q_{f}-\sum_{j\succeq_{f}w}\bar{x}_{f,j}\right]  \cdot\left[  1-\sum_{i\succeq_{w}f}\bar{x}_{i,w}\right]=0.$$

\item[ii)]$w_{1}\succ_{f} w$. We analyse two sub-cases, if the firm $f$ does or does not fill its quota.
\begin{enumerate}
\item[a)] If the firm $f$ does not fill its quota, by Theorem \textcolor{blue}{RHT}, the firm $f$ is assigned to the same set of workers in every stable matching. Assume that $\mu(f)=\{w_{1},\ldots,w_{p} \}$ with $p<q_{f}$. 
Recall that  $w_{1}\succ_{f} w$, then 
$$
0<\sum_{j\succ_{f}w}\bar{x}_{f,j}<q_{f}.
$$

Assume that $0<\sum_{i\succeq_{w}f}\bar{x}_{i,w}<1$, then $\sum_{i\prec_{w}f}\bar{x}_{i,w}>0$. Since $\bar{x}$ is a strongly stable fractional matching for the reduced preference profile $P^\mu$, then by Theorem \ref{sucesion decreciente} there are stable matchings $\mu^1,\ldots,\mu^{k}$ in $(F,W,P^{\mu},\boldsymbol{q})$ and real numbers $\alpha_1,\ldots,\alpha_{k}$, such that $\bar{x}=\sum_{l=1}^{k}\alpha_{l}x^{\mu^{l}}$, with $0<\alpha_{l}\leq 1$, $\sum_{l=1}^{k}\alpha_{l}=1$  and $\mu ^{1}\succ_{F}\mu ^{2}\succ_{F}\ldots\succ_{F}\mu ^{k}$. Since $\sum_{i\prec_{w}f}\bar{x}_{i,w}>0$, then there is a stable matching $\mu^{l}$ for some $l=1,\ldots,k$ such that $\sum_{i\prec_{w}f}x^{\mu^{l}}_{i,w}=1$. Given that $(f,w)\in{A(P)}$, the firm $f$ does not fill its quota, and  $\sum_{i\prec_{w}f}x^{\mu^{l}}_{i,w}=1$, then $\mu^{l}(w)\prec_{w}f$. Hence, $(f,w)$ is a blocking pair for $\mu^{l}$ for some $l=1,\ldots,k,$ and this is a contradiction, then $\sum_{i\succeq_{w}f}\bar{x}_{i,w}=1$. Therefore, $\bar{x}$ fulfils condition (\ref{condition SS(P)}).

\item[b)] If the firm $f$ fill its quota. Without loss of generality, we assume that $\mu(f)=\{w_{1},\ldots,w_{q_f}\}$, $\mu_{W}(f)=\{w'_{1},\ldots,w'_{q_f}\}$, $w_{l}\succ_{f}w_{l+1}$ and $w'_{l}\succ_{f}w'_{l+1}$ for $l=1,\ldots,q_{f}-1$. Notice that, $\mu(f)\cap \mu_{W}(f)$ is not necessarily empty.

We analyse the following 3 sub-cases:
\begin{enumerate}

\item[b$_{1}$)]
$w_{1}\succ_{f} w \succ_{f}w_{q_f}$.

Hence, $\sum_{j\succeq_{f}w}x^{\mu}_{f,j}< q_{f}.$ Since each stable matching fulfils condition (\ref{condition SS(P)}), then $\sum_{i\succeq_{w}f}x^{\mu}_{i,w}=1.$

Since $\bar{x}$ is a stable fractional matching in $P^\mu$, then $x^{\mu}\succeq_{F}\bar{x}$. By Lemma \ref{fuertemente al medio}, we have that $\bar{x}\succeq_{W}x^{\mu}$. Thus,
$$1=\sum_{i\succeq_{w}f}x^{\mu}_{i,w}\leq\sum_{i\succeq_{w}f}\bar{x}_{i,w}\leq 1.$$ Hence, $$\sum_{i\succeq_{w}f}\bar{x}_{i,w}= 1,$$ which implies that $$\left[q_{f}-\sum_{j\succeq_{f}w}\bar{x}_{f,j}\right]\cdot\left[1-\sum_{i\succeq_{w}f}\bar{x}_{i,w}\right]=0.$$
\item[b$_{2}$)]
$w_{q_f}\succ_{f}w\succ_{f}w'_{q_f}$.

Since  $(f,w)\not\in{A(P^{\mu})}$, then $f$ was eliminated from the worker $w$'s preference list $P^{\mu}$. Then, for the worker $w$ we have that $f\succ_{w}\mu_{W}(w)$ or $\mu(w)\succ_{w}f$.
If $f\succ_{w}\mu_{W}$, then the pair $(f,w)$ blocks the matching $\mu_{W}$, then $\mu(w)\succ_{w}f$. Therefore, by Lemma (\ref{fuertemente al medio}) we have that $\bar{x}\succeq_{W}x^{\mu}$. Hence, 
$$
\sum_{i\succeq_{w}f}\bar{x}_{i,w}\geq\sum_{i\succeq_{w}f}x^{\mu}_{i,w}=x^{\mu}_{\mu(w),w}=1.
$$
Since $\bar{x}$ satisfies linear inequality (\ref{ecuacion 2}), we have that $\sum_{i\succeq_{w}f}\bar{x}_{i,w}=1$. Then,
$$
\left[  q_{f}-\sum_{j\succeq_{f}w}\bar{x}_{f,j}\right]  \cdot\left[  1-\sum_{i\succeq_{w}f}\bar{x}_{i,w}\right]=0.
$$
\item[b$_{3}$)]
$w'_{q_f}\succ_{f}w$.

By Lemma (\ref{fuertemente al medio}) we have $\bar{x}\succeq_{F}x^{\mu_{W}}$. Hence,
$$\sum_{j\succeq_{f}w}\bar{x}_{f,j}\geq\sum_{j\succeq_{f}w}x^{\mu_{W}}_{f,j}=q_{f}.$$
Since $\bar{x}$ satisfies linear inequality (\ref{ecuacion 1}), we have that

 $\sum_{j\succeq_{f}w}\bar{x}_{f,j}=q_{f}$. Then,
$$\left[  q_{f}-\sum_{j\succeq_{f}w}\bar{x}_{f,j}\right]  \cdot\left[  1-\sum_{i\succeq_{w}f}\bar{x}_{i,w}\right]=0.$$

\end{enumerate}
\end{enumerate}
\end{enumerate}
\end{enumerate}
From cases 1 and 2, we conclude that $\bar{x}$ is a strongly stable fractional matching for the many-to-one matching market $(F,W,P,\boldsymbol{q})$.\end{proof}

\newpage

\end{document}